\documentclass[10pt]{article}
\usepackage{amsfonts,amsmath,latexsym}
\usepackage[T1]{fontenc}
\usepackage[dvips]{graphicx}
\usepackage{epsfig}
\usepackage{enumerate}
\usepackage{epsf}
\usepackage{amsthm}
\usepackage{color}
\usepackage{amssymb}
\usepackage{fancyhdr} 
\usepackage{hyperref}
\newtheorem{thm}{Theorem}[section]
\newtheorem{propo}[thm]{Proposition}
\newtheorem{lemme}[thm]{Lemma}

\newtheorem{defi}[thm]{Definition}
\newtheorem{remarque}[thm]{Remark}

\newtheorem{Hyp}{Assumption}

\def\R{\mathbb R}
\def\N{\mathbb N}

\def\C{\mathbb C}
\def\E{\mathbb E}

\def\shf{{\cal F}}
\def\shh{{\cal H}}

\def\shl{{\cal L}}

\author{{\sc St\'ephane GOUTTE}\footnote{Universit\'e Paris 13,
    Math\'ematiques LAGA, Institut Galil\'ee, 99 Av. J.B. Cl\'ement 93430
    Villetaneuse. E-mail:{\tt
      goutte@math.jussieu.fr}}\thanks{Luiss Guido Carli - Libera
    Universit\`a Internazionale degli Studi Sociali Guido Carli di Roma}
  {\sc,}\ {\sc Nadia OUDJANE $^*$}\thanks{EDF R\&D,
    Universit\'e Paris 13, FiME (Laboratoire de Finance des
 March\'es de l'Energie (Dauphine, CREST, EDF R\&D) www.fime-lab.org). E-mail:{\tt  
nadia.oudjane@edf.fr} } 
\ {\sc and}\ {\sc Francesco RUSSO }  
\footnote{ENSTA ParisTech, UMA, 
 Unit\'e de Math\'ematiques appliqu\'ees,
32 Bd. Victor,
F-75739 Paris Cedex 15 (France)
}
\thanks{INRIA Rocquencourt 
and Cermics Ecole des Ponts, Projet MATHFI.
  E-mail:{\tt  francesco.russo@ensta-paristech.fr}}
}

\date{October  18th 2011}
\title{\bf Variance Optimal Hedging for discrete time 
processes with independent increments.\\Application to Electricity Markets}

\newcommand{\MBFigure}[6]{
$\left. \right.$ \\
\refstepcounter{figure}
\addcontentsline{lof}{figure}{\numberline{\thefigure}{\ignorespaces #5}}
\begin{center}
\begin{minipage}{#1cm}
\centerline{\includegraphics[width=#2cm,angle=#3]{#4}}
\begin{center}
\upshape{F\textsc{ig} \normal
\end{center}
size{\thefigure}. $-$} #5
\end{center}
\label{#6}
\end{minipage}
\end{center}
$\left. \right.$ \\}

\begin{document}

\maketitle

\begin{abstract}
We consider the discretized version of a (continuous-time) two-factor model introduced by Benth and coauthors
for the electricity markets.
For this model, the underlying is the exponent of a sum of independent
random variables.
We provide and test an algorithm, which is based on the 
celebrated F\"ollmer-Schweizer decomposition for solving the mean-variance
hedging problem. 
In particular, we establish that decomposition explicitly, 
for a large class of vanilla contingent claims.
Interest is devoted in the choice of rebalancing dates and its impact on the
hedging error, regarding the payoff regularity and the non stationarity of the log-price process.
\end{abstract}


\bigskip\noindent  {\bf 2010  AMS-classification}:
60G50,  60G51, 91G10, 60J05, 62M99

\medskip
 \bigskip\noindent   {\bf JEL-classification}: C02, C15, G11, G12, G13

\newpage



\section{Introduction}\label{SectionIntro}

\setcounter{equation}{0}

It is well known that the classical Black-Scholes model does not allow
 in real applications to replicate perfectly contingent claims. Of course, this is due to market incompleteness and  specifically two major reasons~: 
the non-Gaussianity of prices log-returns and the finite number of trading dates. 
The impact of these features have been intensively studied separately in the literature. 
%

There is a large literature on pricing and hedging with non Gaussian models (allowing for stochastic volatility or jumps), in a continuous time setup. 
Then, the hedging error related to the discretization of the hedging strategy is in general ignored or investigated separately. 
One popular approach 
 is the Variance-Optimal hedging. 
Let  $S^c$ denotes the underlying price process where 
the superscript $c$ refers to the \textit{c}ontinuous time setting);
 if $H$ denotes the payoff of the
 option,
the goal is  to minimize the mean squared hedging error
$$
\E[(V_T-H)^2] \quad \textrm{with}\quad V_T=c+\int_0^T  v_t dS^c_t \ .
$$
\noindent
over all initial endowments $c\in \R$ and all (in some sense) 
admissible strategies $v$. 
The first paper specifically on this subject is due to Duffie and
Richardson, see \cite{duffierich}. Among significant early contributions
there are \cite{S94, S95, S01, RS97, GLP98}, a fairly complete recent 
article on the structure of mean-variance hedging,
with a rich bibliography is provided by \cite{CK08}.
One of the now classical tools is the so called F\"ollmer-Schweizer
decomposition. Given a square integrable r.v. $H$   and an
$(\shf_t)$-semimartingale $S = (S_t)_{t \ge 0}$, that 
decomposition consists in finding
a triple $(H_0, \xi, L)$ where $H_0$ is $\shf_0$-measurable,
$\xi$ is $(\shf_t)$-predictable and $L$ is a martingale being
orthogonal to the martingale part $M$ of $S$
such that $H = H_0 + \int_0^T \xi_s dS_s + L_T$.
In the recent years, some attention was focused on finding 
explicit or quasi explicit formulae for the F\"ollmer-Schweizer
decomposition or the optimal strategy for the mean-variance hedging
problem. For instance \cite{nunno}  gave
an expression based on Clark-Ocone type decompositions 
related to L\'evy type measures when the underlying is  a L\'evy
martingale, \cite{ContTankov} still in the martingale case
with techniques of partial integro differential equations.
\cite{Ka06} obtained significant explicit decompositions
when the underlying is the exponential of a L\'evy process
and the contingent claim is a vanilla type 
option appearing as some generalized Laplace transform
of a finite complex measure. 
Other significant  semi-explicit formulae appear in \cite{KP09b, KP09}.
\cite{Ka06} was continued by \cite{GOR}
in the framework of processes with independent increments with some 
applications to the electricity market. 

However, in practice,  the hedging strategy cannot be implemented 
continuously and the resulting optimal strategy has to be  discretized.
 Hence, to be really relevant the hedging error should take into account 
this further approximation. 

\noindent
An alternative approach, less investigated in the literature, is to
consider directly the hedging problem in discrete time as proposed by
Cox Ross
 and Rubinstein~\cite{CRR79}. The first incomplete market analysis
in the spirit of minimizing a quadratic risk is due to \cite{FS89}.
They worked with the so-called local risk-minimization. 
The problem of  Variance-Optimal hedging 
in the discrete time setup 
was proposed in~\cite{Schal94,S95bis}. In the recent years some 
interest on discrete time was rediscovered in 
 \cite{bertsimas, cerny04, KMS09}. \cite{CK09} revisits the seminal paper
\cite{FS89} in the spirit of global risk minimization.
In the discrete-time context, a significant role was played by the
analogous  of the previously mentioned FS-decomposition. It is recalled 
in Definition \ref{def28}.

\noindent
Recently, many approaches have been proposed to obtain explicit or quasi-explicit formulae for computing both the variance optimal trading strategies and hedging errors in discrete time.
For instance, in~\cite{AH10}, Angelini and Herzel 
 derive closed formulae for the variance optimal hedge ratio and the corresponding hedging error variance when the underlying asset is a geometric Brownian motion which is martingale.
As we said, Kallsen and co-authors contributed at providing 
semi-explicit formulae for the Variance-Optimal hedging problem both in
 discrete and continuous time, for various kind of models.
 In particular in~\cite{Ka06}, semi-explicit formula are derived for the
 (discrete and continuous time) Variance-Optimal hedging strategy and
 for the  resulting hedging error, in the specific case where the
 logarithm of the underlying price is a process with 
stationary independent increments. 
One major idea proposed in~\cite{Ka06} 
and~\cite{cerny07} consists in expressing 
the payoff as a linear combination of exponential payoffs for which 
the variance optimal hedging strategy can be expressed explicitly.  
With a similar methodology and in the
same setting, Angelini and Herzel~\cite{AH09} determine the Laplace transform of the
variance of the error produced by a standard delta hedging strategy  when applied to several class of models. In~\cite{DGK09} similar results are provided in the continuous time setup. 
In this paper, we use the generalized Laplace transform approach to
extend   the results of~\cite{Ka06} to the case of processes with
 independent increments (PII) relaxing the stationary assumption 
on  log-returns. 
The semi-explicit discrete F\"ollmer-Schweizer decomposition is stated in 
Proposition \ref{lemme24}, the solution  to the mean-variance hedging
problem in Theorem \ref{ThmSolutionDicret}.
The expression of the quadratic hedging error in Theorem
\ref{thmErrorHedgingDiscret} gives a priori a criterion
of market completeness as far as vanilla options are concerned.
This confirms that the (even not stationary) binomial model
is complete, see Proposition \ref{PBinomial}.

Our discrete time model consists in fact in the discretization
of continuous time models which are exponentials of processes
of independent increments. 
Given a continuous-time model $(S_t^c)_{t \ge 0}$, 
where  $S^c_t = s_0 \exp(X^c_t)$ and $X^c$ is a process with independent
increments and
 discrete trading dates $t_0, t_1,\cdots, t_N$,
our discrete model will be
$S=(S_k)$, such that $S_k=S^c_{t_k}$, for all $k=0,1,\cdots N$.
In this discrete time setting, the Variance-Optimal pricing and hedging problem consists in looking 
for the initial endowments $c\in \R$ and the admissible strategy 
$v=(v_k)$ which minimizes 
$$
\E[(V^N_T-H)^2] \quad \textrm{with}\quad V^N_T=c+\sum_{k=1}^{N}  v_k \Delta S_k \ .
$$
\noindent
This framework is indeed well suited to take into account together
 both the non-Gaussianity of log-returns and hedging errors due to
 the discreteness of trading times. 
Our investigation
for quasi-explicit formulae when the underlying is the
exponential of sums of independent random variables
is due to two reasons.
\begin{enumerate}
\item The first one comes from the fact that the basic continuous
time model can be time-inhomogeneous in a natural way, see for instance
\cite{GOR}.
\item The second, more original reason,  is that the discretized
  times, which correspond in our case to the rebalancing dates, 
are not necessarily  uniformly chosen.
\end{enumerate}
%
About item 1., some prices exhibit non stationary and non-Gaussian log-returns. One common example of this phenomenon can be observed on electricity
 futures or forward market: the forward volatility increases when the time to delivery decreases whereas the tails of log-returns distribution get heavier resulting in huge spikes on the Spot. 
The exponential L\'evy factor model, proposed in~\cite{BS03} and~\cite{oudjaneCollet} allows to represent both the volatility term structure and the spikes on the short term.  More precisely, the forward price given at time $t$ for delivery of 1MWh at time $T_d\geq t$, denoted $F_t^{T_d}$ is then modeled by a two factors model, such that 
\begin{equation}
\label{eq:elecSimple}
S_t^c:= F_t^{T_d}=F_0^{T_d}\exp(m_t^{T_d}+\int_0^t \sigma_S 
e^{-\lambda ({T_d}-u)}
      d\Lambda_u+ \sigma_L W_t)\ ,\quad \textrm{for all}\ t\in [0,{T_d}]\ , 
\end{equation}
where $m$ is a real deterministic trend, $\Lambda$ a real L\'evy process and $W$ a real Brownian motion. 
Hence,  forward prices are modeled as exponentials of PII with \textit{non-stationary increments} and existing results from~\cite{Ka06} valid for stationary independent processes cannot be applied for that kind of models. 

\noindent
Concerning item 2., the  announced motivation for our development is to be able
 to analyze the impact of a non-homogeneous discretization of the trading
 dates on the Variance-Optimal hedging error. 
The issue of considering non-homogeneous trading dates was first considered by Geiss~\cite{Gei02} who analyzed the impact on the hedging error of discretizing a continuously rebalanced hedging portfolio. He showed that for a given irregular payoff (e.g. a digital call), concentrating rebalancing dates near the maturity instead of rebalancing regularly can improve the convergence rate of the hedging error. 
Later, Geiss and Geiss \cite{GeiGei04} introduced  the
 so called {\it fractional smoothness} quantifying the impact
of the payoff irregularity on  the optimal  discretization grid.
 The reader can consult \cite{GeiGo10} for a nice survey on this subject 
and \cite{GM09} for some recent developments.

Hence, it seems to be of real interest to be able to consider such non-homogeneous grids. 
However, if the continuous time log-price model $X^c=\log(S^c) - \log(s_0) 
 $ has independent and stationary increments, considering non-homogeneous
trading dates involves a non stationary discrete time process $X$ such
that $X_k=X^c_{t_{k}}$ for $k=0,\cdots N$, where  $t_0,t_1,\cdots , t_N$
denote the non-homogeneous trading dates. Hence, here again existing
results from~\cite{Ka06} cannot be  applied neither
for hedging at non-homogeneous times nor for evaluating the resulting 
hedging error. 

In the present work,  
we have performed some numerical tests concerning both applications. One major observation is the remarkable robustness of the Black-Scholes strategy that still achieves quasi-minimal hedging errors variances, with  both non Gaussian log-returns and discrete rebalancing dates. 
Besides, our tests show that when hedging with electricity forward contracts, the impact of the choice of the rebalancing dates on the hedging error seems to be more important than the choice of log-returns distribution (Gaussian or Normal Inverse Gaussian, in our case).  
Concerning the case of hedging an irregular payoff (a digital call, in our case), our numerical tests confirm the result of~\cite{Gei02}. In \textit{almost Gaussian cases}, we observe that the variance optimal hedging error, can be noticeably reduced by optimizing the rebalancing dates. However,  this phenomena is less pronounced when the tails of the log-returns distribution get heavier for which the hedging error gets less sensitive to the rebalancing grid. This suggests
 that the result of~\cite{Gei02} and~\cite{GM09} could not be extended straightforwardly to the non Gaussian case. 

%

This article is organized as follows. In Section 2, notations and generalities on the discrete F\"ollmer-Schweizer decomposition are presented. 
In Section 3,  we derive semi-explicit F\"ollmer-Schweizer  decomposition
 for exponential of PII. 
Section 4 is devoted to the solution to the global minimization problem.
Illustrative example and simulation results are given in Section~5;
in particular,  subsection~\ref{electricity} is  concerned 
with data coming from the electricity market.

\section{Generalities and Discrete F\"ollmer-Schweizer
 decomposition}\label{SectionGenera}

We present the context of the problem studied by \cite{S95bis}.

Let $(\Omega,\mathcal{F},P)$ be a probability space, $N \in \mathbb{N}^{*}$ a fixed natural number and $\mathbb{F}=(\mathcal{F}_k)_{k=0,\cdots ,N}$ a 
fixed reference filtration. We shall assume that $\mathcal{F}=\mathcal{F}_N$.
Let $(S_k)_{k=0,\cdots ,N}$ be a real-valued, $\mathbb{F}$-adapted, square-integrable process. We denote by $\Delta S_k$ the increments $S_k-S_{k-1}$, for $k=1,\cdots, N$. We use the convention that a sum (respectively product) over an empty set is zero (resp. one). 
\begin{defi}\label{def21}
We denote by $\Theta$ the set of all predictable processes $v$ (i.e.:
 $v_k$ is $\mathcal{F}_{k-1}$-measurable for each $k \ge 1$) such that
 $v_k\Delta S_k \in \shl^2(\Omega)$ for $k=1,\cdots ,N$. For $v \in \Theta$, $G(v)$ is the process defined by
$$ 
G_k(v):=\sum_{j=1}^k v_j\Delta S_j, \quad {\rm for} \quad \ k=1,\cdots ,N.
$$
\end{defi}
The problem addressed in \cite{S95bis} 
is the following.\\
Given $H\in \mathcal{L}^2(\Omega)$, we look for $(V_0^*,\varphi^*)$
  which minimize the quantity 
\begin{eqnarray}\label{OptPrbDis}\mathbb{E}\left[\left(H-V_0-G_T(\varphi)
\right)^2\right]\ ,
\end{eqnarray}
over  $V_0\in \mathbb{R}$ and  $\varphi \in \Theta$. It will be called \textbf{discrete time optimization problem}. The
expression
$\mathbb{E}\left[\left(H-V_0^*-G_T(\varphi^*)\right)^2\right]$ will be 
 called the \textbf{variance optimal hedging error}.
\begin{defi}\label{def22}
Schweizer \cite{S95bis} introduces the following 
\textbf{non-degeneracy condition (ND)}.
We say that S satisfies the non-degeneracy condition (ND) if there exists a constant $\delta \in ]0,1[$ such that
\begin{eqnarray*}\left(\mathbb{E}[\Delta S_k|\mathcal{F}_{k-1}]\right)^2\leq \delta \mathbb{E}[(\Delta S_k)^2|\mathcal{F}_{k-1}]\ ,\end{eqnarray*}
P.a.s for $k=1,\cdots ,N$.
\end{defi}
\begin{remarque}\label{rem23}
\begin{enumerate}
\item If $(S_k)$ is a
martingale then (ND) is always verified.
\item Note that by Jensen's inequality, we always have
$\left(\mathbb{E}[\Delta
    S_k|\mathcal{F}_{k-1}]\right)^2\leq 
\mathbb{E}[(\Delta S_k)^2|\mathcal{F}_{k-1}] \quad {\rm a.s.}
$
The point of condition (ND) is to ensure a strict inequality
 uniformly in $\omega$.
\end{enumerate}
\end{remarque}
To obtain another formulation of (ND), we now express $S$ in its Doob
decomposition as $S_k=M_k+A_k$
 where $M_k$ is a
square-integrable 
martingale and $A_k$ is a square-integrable predictable
process with $A_0=0$. It is well-known that this decomposition is unique
and is given through
$$
\Delta A_k:=\mathbb{E}[\Delta S_k|\mathcal{F}_{k-1}]\ ,\quad\textrm{and}\quad 
\Delta M_k:=\Delta S_k-\Delta A_k\ .
$$
We will operate with the help of some conditional moments and 
conditional variance setting
$$
Var[\Delta S_k|\mathcal{F}_{k-1}]:=\mathbb{E}[(\Delta S_k)^2|\mathcal{F}_{k-1}]-\mathbb{E}[\Delta S_k|\mathcal{F}_{k-1}]^2\ .
$$
\begin{remarque}\label{remarqueRD}
For $k = 1, \ldots, N$, we have the following.
\begin{enumerate}
\item $\mathbb{E}[(\Delta S_k)^2|\mathcal{F}_{k-1}]=\mathbb{E}[(\Delta M_k)^2|\mathcal{F}_{k-1}]+(\Delta A_k)^2\ ;$
\item $Var[\Delta S_k|\mathcal{F}_{k-1}]=\mathbb{E}[(\Delta M_k)^2|\mathcal{F}_{k-1}]\ ;$
\item Previous conditional variance vanishes if and only if
 $\Delta M_k=0\ $ a.s.
\end{enumerate}
\end{remarque}
We introduce the predictable process $\lambda_k$ by
\begin{equation} 
\label{lambdak}
\lambda_k:=\frac{\Delta A_k}{\mathbb{E}[(\Delta S_k)^2|\mathcal{F}_{k-1}]}=\frac{\mathbb{E}[\Delta S_k|\mathcal{F}_{k-1}]}{\mathbb{E}[(\Delta S_k)^2|\mathcal{F}_{k-1}]}\ ,
\end{equation}
for all $k=1,\cdots ,N$. These quantities could be theoretically infinite.
\begin{remarque}\label{remarque53bis}
Suppose that $P( \Delta S_k=0)=0$ for any $k=1,\cdots ,N$.
\begin{enumerate}
\item Then $\mathbb{E}[(\Delta S_k)^2|\mathcal{F}_{k-1}]>0$ a.s.  
In fact, let $B=\{\omega | \mathbb{E}[(\Delta
S_k)^2(\omega)|\mathcal{F}_{k-1}]
=0\}$. This implies $\Delta A_k=0$ on $B$ because of Remark
\ref{remarqueRD} 1.  By the same Remark,
$$
  0=1_B\mathbb{E}[(\Delta M_k)^2|\mathcal{F}_{k-1}]=
\mathbb{E}[1_B(\Delta M_k)^2|\mathcal{F}_{k-1}]\ ,
$$ 
so $\Delta M_k=0$ a.s. on $B$.  
This implies that $\Delta S_k=0$ a.s. on $B$. By assumption, $B$ is 
forced to be a null set.
\item Previous point 
1. guarantees in particular that $(\lambda_k)$ are all finite.
\end{enumerate}
\end{remarque}
\begin{defi}\label{MVTDis}
The \textbf{mean-variance tradeoff process} of $S$ is defined by
$$
K^d_j:=\sum_{l=1}^{j}\frac{\mathbb{E}[\Delta S_l|\mathcal{F}_{l-1}]^2}{Var[\Delta S_l|\mathcal{F}_{l-1}]}\ ,
$$
for all $j=1,\cdots ,N$. $K^d$ is the discrete version  of
the continuous time corresponding process $K$ defined
for instance in Definition 2.11 of~\cite{GOR} or in Section 1.
 of \cite{S94}.
\end{defi}
\begin{propo}\label{S16}
The condition (ND) is fulfilled if and only if
\begin{eqnarray*}\frac{\mathbb{E}[\Delta S_k|\mathcal{F}_{k-1}]^2}{Var[\Delta S_k|\mathcal{F}_{k-1}]}\end{eqnarray*} 
is a.s. bounded uniformly in $\omega$ and k.
\end{propo}
\begin{proof}
See (1.6) in \cite{S95bis}.
\end{proof}
A basic tool for solving the optimization problem (\ref{OptPrbDis}) in \cite{S95bis} is the discrete F\"ollmer-Schweizer decomposition.
\begin{defi}\label{def28}
Denote by $S=M+A$ the Doob decomposition of $ S$ into a martingale $M$ and a
predictable process $A$.
A complex-valued square integrable random variable $H$ is said to admit
a \textbf{discrete F\"ollmer-Schweizer decomposition} (or simply discrete
FS-decomposition) 
if there exists a  $\mathcal{F}_0$-measurable $H_0$, a complex-valued process
$\xi$ such that both    
$Re \xi(z), Im \xi(z)$ belong to $ \Theta$, and a square integrable
$\C$-valued  martingale $L^H$ such that 
\begin{enumerate}
\item  $L^H M$ is a 
martingale;
\item $E(L_0^H) = 0$, 
\item  $H=H_0+\sum_{k=1}^{N}\xi_k\Delta S_k
+L_N^H$. 
\end{enumerate}
 When Point 1. is fulfilled $L^H$ and $M$ are called {\bf strongly
   orthogonal}.\\
If $H$ is a real valued r.v. then $H$ admits a {\bf real discrete FS
decomposition} if it admits a FS decomposition with $H_0 \in \R$
and $\xi$ being a real valued process. In this case $\xi \in \Theta$.

\end{defi}

\subsection{Existence and structure of an optimal strategy}\label{sectionExictence}

\begin{Hyp}\label{HypND}
$(S_k)_{k=1,\cdots ,N}$ satisfies the non-degeneracy condition (ND).
\end{Hyp}
\begin{remarque} \label{R29}
\begin{enumerate}
\item Under Assumption~\ref{HypND}, Proposition 2.6 of \cite{S95bis} 
guarantees that every square integrable real random variable $H$ admits a
real discrete FS-decomposition.
\item That decomposition is unique because of Remark 4.11 of \cite{Schal94}.
\item The previous two points imply the existence and uniqueness of
the  discrete F\"ollmer-Schweizer decomposition when $H$ is 
a complex square integrable random variable.
\item An immediate consequence is that the decomposition of
a real  square integrable random variable is necessarily real.
\end{enumerate}
\end{remarque}
Other tools for solving the optimization problem and evaluating the
error
 are the following.
\begin{propo}\label{propo210} 
If $S$ satisfies (ND), then $G_N(\Theta)$ is closed in $\shl^2(P)$.
\end{propo}
\begin{proof}
See \cite{S95bis}, Theorem 2.1.
\end{proof}
\begin{thm}\label{T1}
Suppose that $S=M+A$ has a deterministic mean-variance tradeoff
process. Let $H$ be a square integrable real random variable
with discrete real FS- decomposition given by $H=H_0+G_N(\xi^H)+L^H_N$.
\begin{enumerate}
\item The optimization problem (\ref{OptPrbDis}) is solved by 
$(V_0^*,\varphi^*)$ where $V_0^*=H_0$ and $\varphi^*$ is determined by
\begin{eqnarray*}\varphi^*_k=\xi^H_k+\lambda_k(H_{k-1}-H_0-G_{k-1}(\varphi^*)).
\end{eqnarray*} 
\item Suppose that $\shf_0$ is a trivial $\sigma$-field.
The hedging error is given by
\begin{eqnarray*}J_0=
\sum_{k=1}^N\mathbb{E}[(\Delta L_k^H)^2]\prod_{j=k+1}^N
(1-\lambda_j\Delta A_j).\end{eqnarray*} 
\end{enumerate}
\end{thm}
\begin{proof} \
Point 1. follows from Proposition 4.3 of
\cite{S95bis}. Concerning  Point 2., $L_0^H = 0$ a.s.
since $\shf_0$ is  trivial. 
The result follows 
 from Theorem 4.4 of \cite{S95bis};
\end{proof}
 Similarly to \cite{Ka06}, we will calculate it explicitely
 in the case where $S$ is the exponential of process with independent
 increments.

\section{Exponential of PII processes}\label{sectionPII}

From now on, we will suppose 
that $(X_n)_{n=0,\cdots ,N}$ is a sequence of random
variables  with {\bf independent increments},
 i.e. $(X_1-X_0,\cdots ,X_N-X_{N-1})$ are
 independent random variables.
From now on, without restriction of generality, it will not be restrictive 
to suppose $X_0=0$.
We  also define the process $(S_n)_{n=0,\cdots ,N}$ as
 $S_n=s_0 \exp(X_n)$, $0\leq n\leq N$ for some $s_0>0$.

\begin{defi} \label{D31}
We denote $D=\left\{z\in \C \vert \exp(zX_N) \in \shl^1
\right \}$.
\end{defi} 
\subsection{Discrete cumulant generating function}
\begin{defi} \label{D32}
 We define the \textbf{discrete cumulant generating function} as\\
 $m:D\times \{0,\cdots ,N\}\rightarrow \mathbb{C}$ with 
$m(z,n)=\mathbb{E}[e^{z\Delta X_n}]$ for all $n=1,\cdots ,N$ and by convention $m(z,0)\equiv 1$.
\end{defi} 
 This function is a discrete version of the cumulant generating
 function  investigated  in \cite{GOR}.
\begin{remarque} \label{R32}
\begin{enumerate}
\item If $z \in D$ then the property of independent increments implies
  that $m(z,n)=\mathbb{E}[\exp(z \Delta X_n)]$ is well-defined
for all $ z \in D$ and   $n=0,1,\cdots ,N$.
\item If $\gamma \in \mathbb{R}^+\cap D$, Cauchy-Schwarz inequality
  implies
that $[0,\gamma]+i\mathbb{R}\subset D$; 
if $\gamma \in \mathbb{R}^-\cap D$ then $[\gamma,0]+i\mathbb{R}\subset D$.
This shows in particular that $D$ is convex.
\end{enumerate}
\end{remarque}
\begin{remarque}\label{remarque56}
When X has stationary increments then we have $m(z,n)=m(z,1)$ for all
 $n=1,\cdots ,N$. We denote this quantity by $m(z)$ similarly
as in \cite{Ka06}, Section 2.
\end{remarque}
We formulate some assumptions which are analogous to those 
 in  continuous time case, see \cite{GOR}.
\begin{Hyp}\label{HypD1}
\begin{enumerate}
\item $\Delta X_n$ is never deterministic for every $n=1,\cdots ,N$.
\item $2 \in D$.
\end{enumerate}
\end{Hyp}
\begin{remarque}\label{remarque57}
In particular, $S_n\in \shl^2(\Omega)$, for every $n=0,1,\cdots ,N$,
 because $2 \in D$.
\end{remarque}
\begin{lemme}\label{mcont}
 $z\mapsto m(z,n)$ is continuous for any $n=0,1,\cdots ,N$.
In particular, if $K$ is a compact real set then
 $\sup_{z\in K + i \R}|m(z,n)|<\infty$. 
\end{lemme}
\begin{proof}
We set $Y=\Delta X_n$ for fixed $n \in \{1,\cdots ,N\}$. Let $z \in D$
and $(z_p)$ be a sequence converging to z. Obviously $\exp(z_pY)\rightarrow \exp(zY)$ a.s. In order to conclude we need to show that the sequence $(\exp(z_p Y))$ is uniformly integrable. After extraction of subsequences, we can separately
 suppose that
\begin{enumerate}
\item either $\min_{n}Re(z_n)\leq Re(z_p)\leq Re(z)$, for all $p \in \N$,
\item or $\max_n Re(z_n)\geq Re(z_p)\geq Re(z)$, for all $p \in \N$.
\end{enumerate}
This implies the existence of $a,A \in D\cap \R$ such that $a\leq Re(z_p) \leq A$, for all $ p\in \N$.\\
Consequently if $M>0$, for every $p \in \N$, we have  
\begin{eqnarray*}\E[\exp(z_pY)1_{|Y|>M}]\leq \int_{-\infty}^{-M}\exp(yRe(z_p))d\mu_Y(y)+\int_{M}^{\infty}\exp(yRe(z_p))d\mu_Y(y)\end{eqnarray*}
where $\mu_Y$ is the distribution law of $Y$. Previous sum is bounded by
$\int_{-\infty}^{-M}\exp(ay)d\mu_Y(y)+\int_{M}^{\infty}\exp(Ay)d\mu_Y(y)$
Since $M$ is arbitrarily big, the result is established.
\end{proof}
\begin{lemme}\label{lemme521}
Let $n=0,\cdots ,N$.
\begin{enumerate}
\item $\mathbb{E}[e^{\Delta X_n}-1]^2=m(2,n)-2m(1,n)+1$.
\item $Var[e^{\Delta X_n}-1]=m(2,n)-m(1,n)^2$.
\item $\mathbb{E}[e^{\Delta X_n}-1]=m(1,n)-1$.
\end{enumerate}
\end{lemme}
\begin{proof}
Statements 1. and 3. follow in elementary manner using the definition of m.\\
Statement 2. follows from  statement 1.
 and the fact that $\mathbb{E}[e^{\Delta X_n}-1]=m(1,n)-1$.
\end{proof}
\begin{remarque}\label{remarque522}
$m(2,n)-m(1,n)^2$ is strictly positive for any $n=1,\cdots ,N$.
 In fact Assumption
 \ref{HypD1} 1. implies that $e^{\Delta X_n}-1$ is never deterministic.
\end{remarque}

\begin{remarque}\label{SNZ}
For $z \in D$ and   $n\in\{1,\cdots N\} $,  we have $ \mathbb{E} (S_n^z) = s_0^z \prod_{k=1}^n m(z,k).$
\end{remarque}
\begin{propo}\label{propo523}
For $n\in\{1,\cdots N\} $,  we have
\begin{enumerate}
\item $\Delta A_n=\mathbb{E}[\Delta S_n|\mathcal{F}_{n-1}]=(m(1,n)-1)S_{n-1}$.
\item $Var[\Delta S_n|\mathcal{F}_{n-1}]=(m(2,n)-m(1,n)^2)S_{n-1}^2$.
\item Condition (ND) is always satisfied.
\item $$\lambda_n=\frac{1}{S_{n-1}}\frac{m(1,n)-1}{m(2,n)-2m(1,n)+1}.$$
\item The mean-variance tradeoff process $K^d$ is deterministic.
\end{enumerate}
\end{propo}
\begin{proof}
\begin{enumerate}
\item  follows from $\mathbb{E}[\Delta S_n|\mathcal{F}_{n-1}]=
S_{n-1}\mathbb{E}[e^{\Delta X_n}-1]$ and Lemma \ref{lemme521} 3.
\item Since 
\begin{equation} \label{Esq}
\mathbb{E}[(\Delta S_n)^2|\mathcal{F}_{n-1}] = 
 S_{n-1}^2\mathbb{E}[(e^{\Delta X_n}-1)^2], 
\end{equation}
we can write 
\begin{eqnarray*}
Var[\Delta S_n|\mathcal{F}_{n-1}]&:=&\mathbb{E}[(\Delta S_n)^2|\mathcal{F}_{n-1}]-\mathbb{E}[\Delta S_n|\mathcal{F}_{n-1}]^2\ ,\\
&=& S_{n-1}^2\mathbb{E}[(e^{\Delta X_n}-1)^2]-S_{n-1}^2\mathbb{E}[e^{\Delta X_n}-1]^2\,\\
&=& S_{n-1}^2Var[e^{\Delta X_n}-1].
\end{eqnarray*}
The conclusion follows from Lemma \ref{lemme521} 2.
\item We make use of Proposition \ref{S16}. In our context we have
\begin{eqnarray}\frac{\mathbb{E}[\Delta S_n|\mathcal{F}_{n-1}]^2}{Var[\Delta S_n|\mathcal{F}_{n-1}]}=
\frac{(m(1,n)-1)^2}{m(2,n)-m(1,n)^2}.\label{E51}\end{eqnarray}
The denominator of the right-hand side never vanishes because of Remark \ref{remarque522}.
\item It follows from \eqref{lambdak}, \eqref{Esq}, 
Lemma \ref{lemme521} 1. and point 1. of this Proposition.
\item It is a consequence of point 3. and Definition \ref{MVTDis}.
\end{enumerate}
\end{proof}

\subsection{Discrete F\"ollmer-Schweizer decomposition}\label{SectionFSDecom}

Similarly to \cite{Ka06} and \cite{GOR}, we would like to obtain the discrete 
F\"ollmer-Schweizer decomposition of a random variable of the type $H=S_N^z$, for some suitable $z \in \mathbb{C}$. The proposition below generalizes
 Lemma 2.4 of \cite{Ka06}.
\begin{propo}\label{lemme24} Under Assumption \ref{HypD1},
let  $z \in D$ fixed, such that 
$2Re(z)\in D$. Then $H(z)=S^z_N$ admits
a discrete  F\"ollmer-Schweizer decomposition 
\begin{equation*}
\left \{
\begin{array}{ccl}
H(z)_n &=& H(z)_0+\sum_{k=1}^{n}\xi(z)_k \Delta S_k +L(z)_n \\
H(z)_N &=& H(z) = S_N^z
\end{array}
\right.
\end{equation*}
 where 
\begin{eqnarray}\label{Ldiscret}
H(z)_n&=&h(z,n)S_n^z\ , \quad\textrm{for all}\quad n \in \{0, \cdots N \}\nonumber\\
\xi(z)_n&=&g(z,n)h(z,n)S_{n-1}^{z-1}\ ,\quad\textrm{for all}\quad n \in \{1, \cdots N \}\\
L(z)_n&=&H(z)_n-H(z)_0-\sum_{k=1}^{n}\xi(z)_k\Delta S_k \ ,\quad\textrm{for all}\quad n \in \{0, \cdots N \} \nonumber
\end{eqnarray}
and $g(z,n)$, $h(z,n)$ are defined by
\begin{eqnarray}
h(z,n)&:=&\prod_{i=n+1}^{N}\left(m(z,i)-g(z,i)[m(1,i)-1]\right)\label{G4}\\
g(z,n)&:=&\frac{m(z+1,n)-m(1,n)m(z,n)}{m(2,n)-m(1,n)^2}.\label{G5}
\end{eqnarray}
\end{propo}
\begin{remarque} \label{remarque311bis}
\begin{enumerate}
\item $z + 1 \in D$ because $D$ is convex, taking into account 
Assumption \ref{HypD1} 2.
\item If $2 Re (z) $ does not belong to $D$, 
for simplicity, we will set 
$$ g(z,n) \equiv h(z,n) \equiv 
 H(z)_n \equiv \xi(z)_n  \equiv L(z)_n \equiv 0.$$
\item If $K$ is a compact real interval, for any $n \in \{0, \cdots N \}$
 we have
 $\sup_{z \in K + i \R} (\vert g(z,n) \vert + \vert h(z,n)\vert) <
 \infty $.  
\end{enumerate}
\end{remarque}
\begin{remarque}\label{remarque513}
Suppose that  $(X_n)_{n=0,\cdots ,N}$ is a {\bf process with stationary
  increments} i.e. such that   $X_1-X_0,\cdots ,X_N-X_{N-1}$ are
 identically distributed  random variables. \\
 According to Remark \ref{remarque56}, we have 
\begin{eqnarray}g(z,n)=\frac{m(z+1)-m(1)m(z)}{m(2)-m(1)^2}.\label{G5L}
\end{eqnarray}
We will denote in this case $g(z)$ the right-hand side of (\ref{G5L}).
Moreover $h(z,n)=h(z)^{N-n}$ where
\begin{eqnarray}h(z)=m(z)-g(z)[m(1)-1].\label{G4L}\end{eqnarray}
\end{remarque}
\begin{proof}[\textbf{Proof of Proposition \ref{lemme24}}]
Since $z+1\in D$ all the involved expressions are-well defined. Since
$L(z)_0 = 0$, we need to  prove the following.
\begin{enumerate}
\item $L(z)$ is a
square integrable martingale.
\item $L(z) M$ is a
martingale.
\end{enumerate}
From (\ref{Ldiscret}), it follows that
\begin{eqnarray*}\Delta L(z)_n=L(z)_n-L(z)_{n-1}=h(z,n)S_n^z-h(z,n-1)S_{n-1}^z-g(z,n)h(z,n)S_{n-1}^z(e^{\Delta X_n}-1);\end{eqnarray*}
$L(z)_n$ is square integrable for any $n \in \{0,\cdots ,N\}$ since $2z\in
D$ and $(X_n)$ has independent increments. \\ 
Since $S_{n}^z=S_{n-1}^ze^{z\Delta X_n}$, we have
\begin{eqnarray}
\Delta L(z)_n=S_{n-1}^z\left[h(z,n)e^{z\Delta X_n}-h(z,n-1)-g(z,n)h(z,n)(e^{\Delta X_n}-1)\right]\label{DeltaLz}\ ,
\end{eqnarray}
therefore
$\mathbb{E}[\Delta L(z)_n|\mathcal{F}_{n-1}]=S_{n-1}^z\mathbb{E}\left[h(z,n)e^{z\Delta X_n}-h(z,n-1)-g(z,n)h(z,n)(e^{\Delta X_n}-1)\right]$.
\begin{enumerate}
\item To show that $L(z)$ is a martingale it is enough to show that
\begin{eqnarray*}\mathbb{E}\left[h(z,n)e^{z\Delta X_n}-h(z,n-1)-g(z,n)
h(z,n)(e^{\Delta X_n}-1)\right]=0.\end{eqnarray*}
Previous expression is equivalent to   the relation
 $h(z,n)m(z,n)-h(z,n-1)-g(z,n)h(z,n)(m(1,n)-1)=0$ for any $0\leq n\leq N$ which is equivalent to $h(z,n-1)=h(z,n)\left(m(z,n)-g(z,n)(m(1,n)-1)\right)$ for any $0\leq n\leq N$. Previous backward relation with $h(z,N)=1$ leads to $(\ref{G4})$.
\item It remains to prove that $(L(z)_nM_n)$ is a martingale.
Since $L(z)_n$ and $M_n$ are square integrable for any $n $
then $L(z)_n M_n \in \shl^1$. 
  We prove now 
 that $\mathbb{E}[\Delta L(z)_n\Delta M_n|\mathcal{F}_{n-1}]=0$.
Proposition \ref{propo523} 1. implies that the Doob decomposition
 $S=M+A$ of $S$ satisfies $\Delta A_n =(m(1,n)-1)S_{n-1}\ .$
Moreover
\begin{eqnarray*}\Delta M_n=\Delta S_n-\Delta A_n=S_{n-1}(e^{\Delta X_n}-1)-S_{n-1}(m(1,n)-1)=S_{n-1}(e^{\Delta X_n}-m(1,n)).\end{eqnarray*}
Coming back to (\ref{DeltaLz})
\begin{equation*}
\Delta L(z)_n\Delta M_n
= S_{n-1}^{z+1}(e^{\Delta X_n}-m(1,n))\left[h(z,n)
e^{z\Delta X_n}-h(z,n-1)-g(z,n)h(z,n)(e^{\Delta X_n}-1)\right].\end{equation*}
Taking the conditional
 expectation with respect to $\mathcal{F}_{n-1}$, we obtain
\begin{eqnarray*}\mathbb{E}[\Delta L(z)_n\Delta M_n|\mathcal{F}_{n-1}]&=&
\mathbb{E}[S_{n-1}^{z+1}(e^{\Delta X_n}-m(1,n))\\
&&\left[h(z,n)e^{z\Delta X_n}-h(z,n-1)-g(z,n)h(z,n)(e^{\Delta X_n}-1)\right]|\mathcal{F}_{n-1}]\\
&=&S_{n-1}^{z+1}\mathbb{E}[(e^{\Delta X_n}-m(1,n))\\
&&\left[h(z,n)e^{z\Delta X_n}-h(z,n-1)-g(z,n)h(z,n)(e^{\Delta X_n}-1)\right]]\\
&=&S_{n-1}^{z+1}\mathbb{E}[e^{(z+1)\Delta X_n}h(z,n)\\
&-&e^{\Delta X_n}h(z,n-1)-e^{\Delta X_n}g(z,n)h(z,n)(e^{\Delta X_n}-1)\\
&-&m(1,n)h(z,n)e^{z\Delta X_n}+m(1,n)h(z,n-1)\\
&+&m(1,n)g(z,n)h(z,n)(e^{\Delta X_n}-1)].\end{eqnarray*}
Again by Lemma \ref{lemme521}, previous quantity equals zero if and only if 
\begin{eqnarray*}
h(z,n)m(z+1,n)-g(z,n)h(z,n)m(2,n)-m(1,n)h(z,n)m(z,n)+m(1,n)^2g(z,n)h(z,n)=0\ ,
\end{eqnarray*}
or equivalently $m(z+1,n)-g(z,n)m(2,n)-m(1,n)m(z,n)+m(1,n)^2g(z,n)=0$.
Remark \ref{remarque522} finally
 shows that $g(z,n)$ must have the form (\ref{G5}). This concludes
 the proof of Proposition \ref{lemme24}.
\end{enumerate}
\end{proof}

\subsection{Discrete F\"ollmer-Schweizer decomposition of special contingent claims}\label{SectionDFS}

We consider now options $f:\mathbb{C} \rightarrow \mathbb{R}$ as in \cite{GOR} of the type 
\begin{eqnarray}
\label{FORM}H=f(S_N)\ , \quad\textrm{with}\quad f(s)=\int_{\mathbb{C}} s^z\Pi(dz)\ ,
\end{eqnarray}
where $\Pi$ is a (finite) complex measure in the sense of Rudin~\cite{RU87},
  Section~6.1. 
An integral representation of some basic European calls can be found 
in \cite{Ka06} or \cite{GOR}.  \\
The  European Call option  $H=(S_T-K)_+$ and  Put
 option  $H=(K-S_T)_+$ 
  have  a  representation of the form~(\ref{FORM})
provided by the lemma below.
\begin{lemme}\label{Call} Let  $K>0$.
\begin{enumerate}
\item For arbitrary $0<R<1$, $s>0$, we have
\begin{eqnarray}(s-K)_+-s=\frac{1}{2\pi i}
\int_{R-i\infty}^{R+i\infty}s^z\frac{K^{1-z}}{z(z-1)}dz\ .\label{Call2}\end{eqnarray}
\item For  arbitrary $R<0$, $s>0$
\begin{eqnarray}(K-s)_+=\frac{1}{2\pi i}\int_{R-i\infty}^{R+i\infty}s^z\frac{K^{1-z}}{z(z-1)}dz\ .\label{Put1}\end{eqnarray}
\end{enumerate}
\end{lemme}
We need at this point an assumption which depends on the support  of
$\Pi$. We set $I_0 : = {\rm supp \Pi} \cap \R$.
\begin{Hyp}\label{HypD2}
\begin{enumerate}
\item $I_0$ is compact.
\item $ 2 I_0 \subset D$.
\end{enumerate}
\end{Hyp}
\begin{remarque}\label{RVerAss}
\begin{enumerate}
\item Assumption \ref{HypD2} is always verified (for any $ 0<R<1$) 
for the Call since $I_0 = \{R, 1 \} $ is  always included in $[0,1]$
which is a subset of $\frac{D}{2}$ by Assumption \ref{HypD1} 2.
\item  Assumption \ref{HypD2} is  also verified for the Put, 
choosing suitable  $R$
 provided that $D$ contains some negative values.
\end{enumerate}
\end{remarque}
\begin{remarque}\label{RD2}
\begin{enumerate}
\item Since $D$ is convex, Assumption \ref{HypD2} 2. and the fact that
$2 \in D$ imply that $I_0 +1 \subset D$.
\item Since $I_0$ is compact, taking $\Pi = \delta_{z}$ 
for some $z \in \C$, Assumption \ref{HypD2} is equivalent
to the assumptions of Proposition \ref{lemme24}.
\item
Since $I_0$  is compact, Assumption \ref{HypD1} point 1. 
and Lemma \ref{mcont} imply that
$\sup_{z \in 2 I_0 + i \R}|m(z,n)|<\infty$, for every $n=1,\cdots ,N$.
\item Taking into account Remark \ref{remarque311bis} and points 2. and 3.
 we also get
$\sup_{z \in \C}(|g(z,n)| + |h(z,n)|) <\infty$, for every $n=1,\cdots ,N$.
\end{enumerate}
\end{remarque}
\begin{remarque}\label{RD3}
Notice that  Assumption~\ref{HypD2} is relatively weak and verified for a large class of models, whereas Assumption~8 required in~\cite{GOR} to derive similar results, in the continuous time setting, noticeably restricts the set of underlying dynamics. 
\end{remarque}
\begin{lemme}\label{remarque514}
For any $n\in\{0,\cdots , N\}$, according to the notations of Proposition
\ref{lemme24} we have
\begin{enumerate}
\item   $\sup_{z \in \C} \mathbb{E}[|H(z)_n|^2] < \infty$; 
\item $\sup_{z \in \C}\mathbb{E}[|\xi(z)_n|^2 (\Delta S_n)^2] < \infty$, for $n\geq 1$;
\item $\sup_{z \in \C}\mathbb{E}[(\Delta L(z)_n)^2] < \infty$.
\end{enumerate}
\end{lemme}
\begin{proof} 
Remark \ref{remarque57}, together with point 4. of Remark \ref{RD2}
show the validity of point 1. Point 3. is a consequence of points 1 and 2.
Concerning this last point, let $n \in \{1,\cdots, N\}$. By Lemma
\ref{lemme521} 1.
\begin{eqnarray*}
 \mathbb{E}[|\xi^H(z)_n|^2 (\Delta S_n)^2] &=&
g(z,n)^2 h(z,n)^2 \E(S^{2z}_{n-1}) (m(2,n)-2m(1,n) +1) \\
&=& g(z,n)^2 h(z,n)^2 m(2z,n-1) (m(2,n)-2m(1,n) +1)
\end{eqnarray*}
The conclusion follows by Remark \ref{RD2}.
\end{proof}
Proposition below extends Proposition 2.5 of \cite{Ka06}.
\begin{propo}\label{propoHdiscret}
We suppose the validity of Assumptions \ref{HypD1} and  \ref{HypD2}.
Any contingent claim $H=f(S_N)$ admits the real discrete
FS decomposition $H$ given by
\begin{equation*} 
\left \{ 
\begin{array}{ccl}
H_n &=&  H_0+\sum_{k=1}^{n}\xi_k^H\Delta S_k + L^H_n \\
H_N &=& H 
\end{array}
\right.
\end{equation*}
 where
\begin{eqnarray}H_n&=&\int_\mathbb{C} H(z)_n \Pi(dz)\\
\xi^H_n  &=&\int_\mathbb{C} \xi(z)_n \Pi(dz)\\
L^H_n&=&\int_\mathbb{C} L (z)_n\Pi(dz)=H_n-H_0-\sum_{k=1}^{n}
\xi_k^H \Delta S_k, 
\end{eqnarray}
according to the same notations as in Proposition \ref{lemme24}
and Remark \ref{remarque311bis}. Moreover the processes 
$(H_n)$,$(\xi^H_n)$ and $(L^H_n)$ are real-valued.
\end{propo}
\begin{proof}
We proceed similarly to \cite{Ka06}, Proposition 2.1.
We need to prove that $L^H$ (resp. $L^HM$) is a square integrable 
(resp. integrable)  martingale.
This will follow from Proposition \ref{lemme24} and Fubini's theorem.
The use of Fubini's is justified by Lemma 
\ref{remarque514}. 
The fact that $ H, \xi^H$ and $L^H$ are real processes follows from 
Remark \ref{R29} 4. 
\end{proof}

\section{The solution of the minimization problem}\label{SectionSolution}

\subsection{Mean-Variance Hedging}\label{subsectionMVH}
We can now summarize the solution to the optimization problem.
\begin{thm}\label{ThmSolutionDicret}
We suppose the validity of Assumptions \ref{HypD1} and  \ref{HypD2}.
Let $H=f(S_N)$ with discrete real FS-decomposition
\begin{equation*} 
\left \{ 
\begin{array}{ccl}
H_n &=&  H_0+\sum_{k=1}^{n}\xi_k^H\Delta S_k + L^H_n \\
H_N &=& H. 
\end{array}
\right.
\end{equation*}
A solution to the optimal problem $(\ref{OptPrbDis})$ is given by
 $(V_0^*,\varphi^*)$ with $V_0^*=H_0$ and $\varphi^*$ is determined by
\begin{eqnarray}\varphi^*_n=\xi^H_n+\lambda_n  \left(H_{n-1}
-H_0-\sum_{i=1}^{n-1}\varphi_i^*\Delta S_i\right)\end{eqnarray}
where $\lambda_n$ is defined for all $n\in\{1,\cdots , N\}$, by
\begin{eqnarray}
\lambda_n=\frac{1}{S_{n-1}}\frac{m(1,n)-1}
{m(2,n)-2m(1,n)+1}\ . 
\end{eqnarray}
Moreover the solution  is unique (up to a null set).
\end{thm}
\begin{remarque}\label{rem42}
In the case that X has stationary increments, we obtain
$$
\lambda_n=\frac{1}{S_{n-1}}\frac{m(1)-1}{m(2)-
2m(1)+1}\ ,
$$
where $m(n) = E(\exp(n X_1))$.
This confirms the results of Section 2. in \cite{Ka06}.
\end{remarque}
\begin{proof}[\textbf{Proof of theorem \ref{ThmSolutionDicret}}]
The existence follows from Theorem \ref{T1}, Proposition
\ref{propoHdiscret} and Proposition \ref{propo523} points 3., 4. and 5.\\
Uniqueness follows exactly as in the proof of Proposition 2.5 of
 \cite{Ka06}: in our case Lemma \ref{lemme521}
 gives $$Var[e^{\Delta X_n}-1]=m(2,n)-m(1,n)^2.$$

\end{proof}

\subsection{The Hedging Error}\label{subsecHerr}

The hedging error is given by
 Theorem \ref{T1} since the mean-tradeoff process is deterministic.
\begin{thm}\label{thmErrorHedgingDiscret}
We suppose the validity of Assumptions \ref{HypD1} and  \ref{HypD2}.
The variance of the hedging error in Theorem \ref{ThmSolutionDicret}
equals
\begin{eqnarray}
\label{eq:J0}J_0= 
\int_\C\int_\C J_0(y,z)\Pi(dy)\Pi(dz)\ ,
\end{eqnarray}
with
\begin{equation}\label{eq:J02}J_0(y,z)= 
\left \{ 
\begin{array}{ccc}
s_0^{y+z} \sum_{k=1}^N  b(y,z;k)h(z,k)h(y,k)
\prod_{\ell =2}^k m(y+z, \ell -1)
\prod_{j=k+1}^Na(j) &:& y, z \in {\rm supp \pi} \\
0 &:& {\rm otherwise}
\end{array}
\right.
\end{equation}
where
\begin{eqnarray*}a(j)=\frac{m(2,j)-m(1,j)^2}{m(2,j)-2m(1,j)+1}\end{eqnarray*}
and 
\begin{equation} \label{FORMb}
b(y,z;k) = \frac{\rho(y,z;k) \rho(1,1;k) - \rho(y,1;k) \rho(z,1;k)}{\rho(1,1;k)}\ ,
 \end{equation}
where $ \rho(y,z;k) = m(y+z,k) - m(y,k)m(z,k)$,
$y, z \in {\rm supp \Pi}$.
\end{thm}
\begin{remarque} \label{RRho}
The function $\rho$ above plays an analogous role to the
complex valued function with the same name introduced in \cite{GOR}
at Definition 4.3 in the continuous time framework.
\end{remarque}
\begin{proof}
We proceed again similarly to the proof of theorem 2.1 of \cite{Ka06}. 
Theorem \ref{T1} gives that the hedging error is given by
\begin{eqnarray} \label{ExprJ0}
J_0 = \sum_{k=1}^N\mathbb{E}[(\Delta L_k^H)^2]
\prod_{j=k+1}^N(1-\lambda_j\Delta A_j)\ .\end{eqnarray} 
Proposition \ref{propo523} gives
\begin{eqnarray} \label{EDelta9}
\Delta A_j &=& \mathbb{E}[\Delta S_j|\mathcal{F}_{j-1}]=(m(1,j)-1)S_{j-1} 
\nonumber\\
&& \\
\lambda_j &=& \frac{1}{S_{j-1}}\frac{m(1,j)-1}{m(2,j)-2m(1,j)+1}, \nonumber
\end{eqnarray}
so 
\begin{equation} \label{EDelta7}
 1 - \lambda_j \Delta A_j = a(j),
\end{equation} 
 and it remains to calculate
$\mathbb{E}[(\Delta L_k^H)^2]$.
Since
\begin{eqnarray*}\Delta L_k^H  =  \int_\mathbb{C} 
\Delta L(z)_k\Pi(dz)\end{eqnarray*}
we have
\begin{equation} \label{EDelta8}
(\Delta L_k^H)^2  =  \int_\mathbb{C}\int_\mathbb{C}
 \Delta L(y)_k \Delta L(z)_k\Pi(dy)\Pi(dz)\end{equation}
and hence by Fubini's Theorem
\begin{eqnarray*}\mathbb{E}[(\Delta L_k^H)^2]  = 
 \int_\mathbb{C}\int_\mathbb{C} \mathbb{E}[\Delta L(y)_k \Delta L(z)_k]
\Pi(dy)\Pi(dz).\end{eqnarray*}
Relation (\ref{DeltaLz}) says that
\begin{eqnarray*}\Delta L(z)_k &=& S_{k-1}^{y+z}
\left[h(y,k)e^{y\Delta X_k}-h(y,k-1)-g(y,k)h(y,k)
(e^{\Delta X_k}-1)\right] \\
&& \left[h(z,k)e^{z\Delta X_k}-h(z,k-1)-g(z,k)h(z,k)
(e^{\Delta X_k}-1)\right].
\end{eqnarray*}
Taking the expectation we obtain
\begin{eqnarray*}
\mathbb{E}[\Delta L(y)_k\Delta L(z)_k] &=& \mathbb{E}[S_{k-1}^{y+z}]
\{(h(z,k)h(y,k)m(y+z,k)-h(z,k)h(y,k-1)m(z,k)\\
&-&h(z,k)h(y,k)g(y,k)\mathbb{E}[e^{z\Delta X_k}(e^{\Delta X_k}-1)]-h(z,k-1)h(y,k)m(y,k)\\
&+&h(z,k-1)h(y,k-1)+h(z,k-1)h(y,k)g(y,k)\mathbb{E}[e^{\Delta X_k}-1]\\
&-&h(z,k)h(y,k)g(z,k)\mathbb{E}[e^{y\Delta X_k}(e^{\Delta X_k}-1)]+h(z,k)h(y,k-1)g(z,k)\mathbb{E}[e^{\Delta X_k}-1]\\
&+&h(z,k)h(y,k)g(z,k)g(y,k)\mathbb{E}[(e^{\Delta X_k}-1)^2]\}.
\end{eqnarray*}
Recalling that $\mathbb{E}[(e^{\Delta X_k}-1)^2]=m(2,k)-2m(1,k)+1$ and $\mathbb{E}[e^{\Delta X_k}-1]=m(1,k)-1$, we obtain
\begin{eqnarray} \label{ExprDelta}
\mathbb{E}[\Delta L(y)_k\Delta L(z)_k] &=&
\mathbb{E}[S_{k-1}^{y+z}]\{(h(z,k)h(y,k)m(y+z,k)-h(z,k)h(y,k-1)m(z,k) \nonumber\\
&-&h(z,k)h(y,k)g(y,k)(m(z+1,k)-m(z,k))-h(z,k-1)h(y,k)m(y,k) \nonumber\\
&+&h(z,k-1)h(y,k-1)+h(z,k-1)h(y,k)g(y,k)(m(1,k)-1) \nonumber\\
&& \\
&-&h(z,k)h(y,k)g(z,n)(m(y+1,k)-m(y,k)) \nonumber\\
&+&h(z,k)h(y,k-1)g(z,k)(m(1,n)-1) \nonumber \\
&+&h(z,k)h(y,k)g(z,k)g(y,k)(m(2,k)-2m(1,k)+1)\}. \nonumber
\end{eqnarray}
By Proposition \ref{lemme24} we have for $x=y$ or $z$ that
\begin{eqnarray} \label{Exprk}
h(x,k-1)&=&h(x,k)[m(x,k)-g(x,k)(m(1,k)-1)].
\end{eqnarray}
We replace the right-hand sides of \eqref{Exprk} in \eqref{ExprDelta}
 and we factorize  by $h(z,k)h(y,k)$. Finally, after simplification
 we obtain
\begin{eqnarray*}
\mathbb{E}[\Delta L(y)_k\Delta L(z)_k] &=& 
\mathbb{E}[S_{k-1}^{y+z}]h(z,k)h(y,k)\{m(y+z,k)\\
&-&m(z,k)m(y,k)+m(z,k)g(y,k)m(1,k)+m(y,k)g(z,k)m(1,k)\\
&-&g(y,k)m(z+1,k)-g(z,k)m(y+1,k)\\
&-&g(z,k)g(y,k)[m(1,k)-1]^2\\
&+&g(z,k)g(y,k)[m(2,k)-2m(1,k)+1]\}.
\end{eqnarray*}
Hence,
\begin{equation} \label{EDelta10}
\mathbb{E}[\Delta L(y)_k\Delta L(z)_k]
 = \mathbb{E}[S_{k-1}^{y+z}]h(z,k)h(y,k) \tilde b(y,z;k),
\end{equation}
where
 \begin{equation} \label{EDelta11}
\mathbb{E}[S_{k-1}^{y+z}]=
s_0^{y+z}\mathbb{E}[e^{(y+z)\Delta X_{k-1}}]=s_0^{y+z}
\prod_{\ell =2}^k m(y+z, \ell -1)
\end{equation}
and
 \begin{eqnarray*} \tilde b(y,z;k)&=&\{m(y+z,k)-m(z,k)m(y,k)-
g(y,k)m(z+1,k)-g(z,k)m(y+1,k)\\
&+&m(z,k)g(y,k)m(1,k)+m(y,k)g(z,k)m(1,k)\\
&-&g(z,k)g(y,k)m(1,k)^2+g(z,k)g(y,k)m(2,k)\}.
\end{eqnarray*}
We observe that
 \begin{equation} \label{EDelta12}
\tilde b(y,z;k)=\rho(y,z;k)-g(y,k)\rho(z,1;k)
 -g(z,k)\rho(y,1;k)+g(y,k)g(z,k)\rho(1,1;k).\end{equation}
Since, for $x=y$ or $z$ 
 \begin{eqnarray*}g(x,k)=\frac{\rho(x,1;k)}{\rho(1,1;k)} \end{eqnarray*}
it follows that $\tilde b(y,z;k)= b(y,z;k)$.
Finally, \eqref{EDelta9}, \eqref{EDelta7},
\eqref{EDelta8}, \eqref{EDelta10},
 \eqref{EDelta11} and \eqref{EDelta12}
give
\begin{eqnarray*}J_0(y,z) &=& s_0^{y+z} 
\sum_{k=1}^N   b(y,z;k) h(z,k)h(y,k)
\prod_{\ell =2}^k m(y+z, \ell -1) 
\prod_{j=k+1}^N\left(1-\lambda_j\Delta A_j\right)\\
&=&s_0^{y+z}  \sum_{k=1}^N   b(y,z;k)h(z,k)h(y,k)
\prod_{\ell =2}^k m(y+z, \ell -1)
\prod_{j=k+1}^N a(j).
\end{eqnarray*}
\end{proof}
From the expression of the variance of the
 hedging error  ~\eqref{eq:J02}, we can derive a sort of criterion for completeness for market asset pricing 
models. More precisely, the condition
\begin{equation}
\label{eq:complete}
b(y,z;k)=0\ ,\quad\textrm{for all}\ y,z\in D\ \textrm{and}\ k\in \{1,\cdots N\}
\end{equation}
characterizes the prices models that are exponential of PII  for which 
every payoff (that can be written as an inverse Laplace transform) can be hedged. 
In the specific case of a Binomial (even inhomogeneous) model, we retrieve the fact that $J_0(y,z)\equiv 0$ and so $J_0=0$.
In fact, that model is complete.
\begin{propo}\label{PBinomial}
Let $a,b \in \R$, $X_k=a$ with probability $p_k$ and $X_k=b$ with probability $(1-p_k)$. Then $J_0(y,z)\equiv 0$ for every $y,z \in \frac{D}{2}$.
\end{propo}
\begin{proof}
Writing $p=p_k$, $k \in \{0,1,\cdots,N\}$, we have
\begin{eqnarray*}\rho(y,z;k)&=&pe^{a(y+z)}+(1-p)e^{b(y+z)}-(pe^{ay}+(1-p)e^{by})(pe^{az}+(1-p)e^{bz})\\
&=&p(1-p)\left(e^{a(y+z)}+e^{b(y+z)}+e^{by+az}+e^{bz+ay}\right)\\
&=&p(1-p)\left(e^{ay}+e^{by}\right)\left(e^{az}+e^{bz}\right).\end{eqnarray*}
So $\rho(y,z;k)\rho(1,1;k)=p^2(1-p)^2\left(e^{ay}+e^{by}\right)\left(e^{az}+e^{bz}\right)\left(e^{a}+e^{b}\right)^2$.
On the other hand, this obviously equals $\rho(y,1;k)\rho(z,1;k)$.
\end{proof}
If $X$ is a  process with stationary  and independent increments
we reobtain the result of \cite{Ka06}].
\begin{propo}\label{rem44}
Let$(X_k)$ be a process with stationary increments.
We denote 
\begin{eqnarray*}
m(y) &:=& \E(\exp(y X_1))   \\
\shh(y) &:=& m(y) - \frac{m(1)-1}{m(2)-m(1)^2} (m(y+1) - m(1) m(y))\\
 a&:=& \frac{m(2)-m(1)^2}{m(2) - 2m(1) + 1}.
\end{eqnarray*}
Then 
\begin{eqnarray*}J_0=\int_\C\int_\C J_0(y,z)\Pi(dy)\Pi(dz)\end{eqnarray*}
with
\begin{equation} J_0(y,z)=
\left \{ 
\begin{array}{ccc}
 s_0^{y+z}  \beta(y,z)\frac{a(y,z)^N-m(y+z)^N}{a(y,z)-m(y+z)}, & if & a(y,z) 
\neq m(y+z)\\
s_0^{y+z} \beta(y,z)Nm(y+z)^{N-1} &if & {\rm a(y,z) = m(y+z)}
\end{array},
\right.
\end{equation}
where
\begin{eqnarray*}
a(y,z)&=& a \shh(y) \shh(z), \\
 \beta(y,z)&=& m(y+z)-\frac{m(2)m(y)m(z)-m(1)m(y+1)m(z)-
m(1)m(y)m(z+1)+m(y+1)m(z+1)}{m(2)-m(1)^2}.\end{eqnarray*}
 \end{propo}
\begin{proof} \
We observe that for $k \in \{0,\cdots,N\}$, we have
\begin{equation*}m(y+z,k)=m(y+z) , \quad h(y,k)=\shh(y)^{N-k}  \quad \textrm{and} \quad h(z,k)=\shh(z)^{N-k}
\end{equation*}
So
\begin{eqnarray*}\prod_{j=k+1}^Na(j)=\left(\frac{m(2)-m(1)^2}{m(2)-2m(1)+1}
\right)^{N-k}=a^{N-k}. \end{eqnarray*}
Consequently, expression \eqref{eq:J02} for $y,z\in supp(\Pi)$,
\begin{eqnarray*}J_0(y,z)=s_0^{y+z}\beta(y,z)\sum_{k=1}^{N}m(y+z)^{k-1}\left(\shh(y)\shh(z)a\right)^{N-k}\end{eqnarray*}
\begin{equation}J_0(y,z)= 
\left \{ 
\begin{array}{ccc}
  s_0^{y+z}\beta(y,z)\frac{(m(y+z)-\shh(y)\shh(z)a)^N}{m(y+z)-a
\shh(y)\shh(z)} & {\rm if}& m(y+z)\neq a\shh(y)\shh(z) \\
s_0^{y+z}\beta(y,z)Nm(y+z)^{N-1} & {\rm if}& m(y+z)=a\shh(y)\shh(z)
\end{array}.
\right.
\end{equation}
This concludes the proof of the proposition.
\end{proof}

\section{Numerical results}\label{secSimu}
As announced in the introduction, we will now apply the quasi-explicit
 formulae derived in previous sections to measure the impact of the choice
 of the  rebalancing dates on the hedging error. 
We will consider two cases that motivated the present work: 
\begin{enumerate}
	\item the underlying continuous time log-price model has stationary increments but the payoff to hedge
 is irregular, such as  a \textbf{Digital call},  so that, as shown in~\cite{Gei02,GM09}, 
 hedging near the maturity can improve the hedge; 
	\item the payoff is regular (e.g. classical call) but
 the underlying continuous time model shows a volatility term structure
  which is exponentially increasing near the maturity, such as
 \textbf{electricity forward prices}.
For this reason it seems  again judicious to 
 hedge more frequently near the maturity, where the volatility accelerates. 
\end{enumerate}

\subsection{The case of a Digital option}

We consider the problem of hedging and pricing a Digital call, with payoff $f(s)=\mathbf{1}_{[K,\infty)}(s)$
of maturity $T > 0$. 
From~(35) in~\cite{Ka06}, the payoff of this option can be expressed as
\begin{eqnarray}
f(s)=\lim_{c\rightarrow\infty}\frac{1}{2\pi i}\int_{R-ic}^{R+ic}s^z\frac{K^{-z}}{z}dz\label{Digital1}\ ,
\end{eqnarray}
for an arbitrary $R>0$.
 This implies that the complex measure $\Pi$ is formally given by
\begin{eqnarray}
\Pi(dz)=\frac{1}{2\pi i}\frac{K^{-z}}{z} \delta_R(dRe (z)) d(i (Im (z))).
\end{eqnarray}
However, such measure is only $\sigma$-finite so that application 
of Theorem~\ref{ThmSolutionDicret} is not rigorously valid. Nevertheless, 
using improper integrals one is able to recover an exploitable
form for applications.

\begin{propo}\label{PDigital}
Let $f(s)=1_{[K,\infty[}(s)$. We suppose again the validity of Assumption \ref{HypD1} and \ref{HypD2}. Let $R>0$.
\begin{enumerate}
\item The FS-decomposition of the contingent claim $H=f(S_N)$ is given by 
\begin{equation*} 
\left \{ 
\begin{array}{ccl}
H_n &=&  H_0+\sum_{k=1}^{n}\xi_k^H\Delta S_k + L^H_n \\
H_N &=& H 
\end{array}
\right.
\end{equation*}
 where
\begin{eqnarray}H_n&=&\lim_{\ell\rightarrow \infty}\int_{R+i[-\ell,\ell]} H(z)_n \Pi(dz)\\
\xi^H_n  &=&\lim_{\ell \rightarrow \infty}\int_{R+i[-\ell,\ell]}  \xi(z)_n \Pi(dz)\\
L^H_n&=&\int_\mathbb{C} L (z)_n\Pi(dz)=H_n-H_0-\sum_{k=1}^{n}
\xi_k^H \Delta S_k, 
\end{eqnarray}
according to the same notations as in Proposition \ref{lemme24}
and Remark \ref{remarque311bis}.
\item The solution to the minimization problem is still given by Theorem \ref{ThmSolutionDicret}.
\item The variance of the hedging error is given by
$$
\lim_{\ell\rightarrow \infty}\int_\C\int_\C J_0(y,z)\Pi_\ell(dy)\Pi_\ell(dz),
$$
where for each $\ell > 0$, $ \Pi_\ell$ is the finite
complex measure defined by 
$\Pi_\ell(B)=\Pi(B\bigcap\left(R+i[-\ell,\ell]\right))$
for a Borel set $B \subset \C$. 
\end{enumerate}
\end{propo}
\begin{proof}
We proceed similarly as in Lemma 4.2 of \cite{Ka06}. 
For $\ell >0$, we denote $f_\ell:\C\rightarrow \C$ defined by
 $f_\ell(s) := \int_\C s^z\Pi_\ell(dz)$. According to the proof of Lemma 4.2 of \cite{Ka06}, there is $u\in \R$ such that 
$$
\vert f(s)-f_\ell(s)\vert \leq us^R, \quad \forall s \in \R
$$
The Lebesgue's dominated convergence theorem implies that $\lim_{\ell\rightarrow \infty}\E\left[(f_\ell(S_N)-f(S_N))^2\right]=0$. Setting $H^\ell=f_\ell(S_N),H=f(S_N)$ we get $\lim_{\ell\rightarrow \infty}\E\left[(H-H^\ell)^2\right]=0$. Item 1. follows by Proposition \ref{AP2}. \\
Item 2. follows by the same arguments as the proof of Theorem \ref{ThmSolutionDicret}.\\
Item 3. follows exactly as in step 3. of the proof of Lemma 4.2 of \cite{Ka06}.
\end{proof}

 In this section, this will be assumed so that 
formula~\eqref{eq:J0} will be used in the case of a Digital option. \\
The underlying process $S^c$ is given as the exponential of a Normal 
Inverse Gaussian L\'evy process (see Appendix~\ref{Appendix:NIG}. B. i.e.
 for all $t\in [0,T]$, 
$$
S_t^c=e^{X_t^c}\ ,\quad \textrm{where $X^c$ is a L\'evy process with } \ 
 X^c_1\,\sim\, NIG(\alpha,\beta,\delta,\mu)\ .
$$
Given $N+1$ discrete dates $0=t_0<t_1<\cdots <t_N=T$,
we associate the discrete model pricing $X = X^N$
where $X_k = X^c_{t_k}, k \in \{ 0, \ldots, N\} $.
$X$ is a  discrete time process with independent increments.
The related cumulant generating 
function $z\mapsto m(z,k)$ 
associated to the increment 
$\Delta X_k=  X_{k}-X_{k-1} =   X^c_{t_k}-X^c_{t_{k-1}}$ for $k\in\{1,\cdots N\}$ 
is defined on  $ D =  [-\alpha-\beta;\alpha-\beta] $.
 We refer for this  to
\cite{GOR} Remark 3.21 2., since $ X^c$ is a NIG process.
By additivity we can show that
\begin{equation} \label{q:mNIG}
 m(z,k)=\mathbb{E} [\exp(z\Delta X_k)]=
\exp{\left (\Delta t_k \big [\mu z +\delta 
\big (\sqrt{\alpha^2-\beta^2}-\sqrt{\alpha^2-
(\beta+z)^2}\big )\big ]\right )}\ , \quad \textrm{for} \ z \in D, k \in \{0, \ldots, N\}.
\end{equation}
For other informations on the NIG law, the reader can refer to 
 Appendix \ref{Appendix:NIG} B. \\
Assumption~\ref{HypD1} 1. is trivially verified, 
  Assumption~\ref{HypD1} 2. is verified as soon as $2 \le \alpha-\beta$. 
 Thanks to Remark~\ref{RVerAss} Assumption~\ref{HypD2} is automatically
 verified for the Call and Put representations 
 given by Lemma~\ref{Call}, and, by similar arguments,
 even for the digital option.  \\
The time unit is the year and the interest rate is zero in all our tests.
 The initial value of the underlying is $s_0=100$ Euros. 
 The maturity of the option is $T=0.25$ i.e. three months from now. 
Four different sets of parameters for the NIG distribution have been
 considered, going from the case of \textit{almost Gaussian}
 returns corresponding to standard equities, 
 to the case of \textit {highly non Gaussian} returns. 
The standard set of parameters 
is  estimated on the \textit{Month-ahead} \textit{base} forward
 prices of the French Power market in 2007:  
\begin{equation}
\label{eq:para;levy}
\alpha=38.46 \,,\  \beta= -3.85\,,\ \delta = 6.40\,,\  \mu= 0.64 \ .
\end{equation}
Those parameters imply a zero mean, a standard deviation of $41\%$, a
skewness (measuring the asymmetry) of $-0.02$ and an excess kurtosis
(measuring the \textit{fatness} of the tails) of $0.01$. The other sets
of parameters are obtained by multiplying parameter  $\alpha$ by a
coefficient $C$, ($\beta,\delta,\mu$) being such that the first three
moments are unchanged. Note that when $C$ grows to infinity the tails of
the NIG distribution get closer to the tails of the Gaussian
distribution. For instance, Table~\ref{tab:kurtosis} shows how the
excess kurtosis (which is zero for a Gaussian distribution) is modified
with the four values of $C$ chosen in our tests. 
\begin{table}[htbp]
\begin{center}
\begin{tabular}{|c||c|c|c|c|c|}
\hline
 \textrm{Coefficient} &  $C=0.14$ & $C=0.2$ & $C=1$ &$C=2$ \\
\hline
\hline
 $\alpha$&    $5.38$&  $7.69$&  $38.46$&  $76.92$ \\
\hline
 Excess kurtosis&   $0.61$&  $0.30$&  $0.01$&  $4. \,10^{-3}$ \\
\hline
\end{tabular}
\end{center}
\caption{{\small Excess kurtosis of $X_1$ for different values of $\alpha$,  $(\beta, \delta, \mu)$ insuring the same three first moments.}}
\label{tab:kurtosis}
\end{table}
We compute the Variance Optimal (VO) hedging error given by~\eqref{eq:J0},  for different grids of rebalancing dates.
The corresponding initial capital $V_0$ denoted by $V_0^\ast = H_0$
in Theorem \ref{ThmSolutionDicret} is computed using Proposition
\ref{propoHdiscret}. 

 In particular, we consider   
the parametric grid introduced in~\cite{Gei02} and~\cite{GM09}
 $\pi^{b,N}:=\{0=t_0^{b,N},t_1^{b,N},\cdots ,t_N^{b,N}\}$ defining, for any real $b\in(0,1]$, $N$ rebalancing dates such that 
\begin{equation}
\label{eq:pib}
t_k^{b,N}=T-T(1-\frac{k}{N})^{1/b}\,\quad\textrm{for all}\ k\in \{0,\cdots , N-1\}\ .
\end{equation}
Note that $\pi^{1,N}$ coincides with equidistant rebalancing dates
 whereas when $b$ converges to zero, the rebalancing dates
 concentrate near the maturity. 
To visualize the impact of parameter $b$ on the rebalancing dates
 grid, we have reported on Figure~\ref{fig:b} the sequences of
 rebalancing dates generated by $\pi^b_N$ for different values of $b$. 
\begin{figure}[htbp]
\begin{center}`
   \epsfig{figure=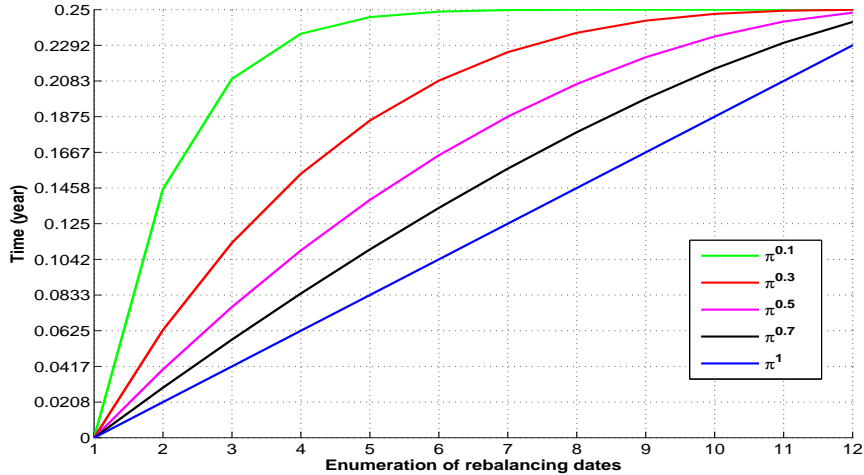,height=7cm, width=13cm}
\end{center}
\caption{{\small Sequences of rebalancing dates for different values of $b$, for $N=12$. }}
\label{fig:b}
\end{figure}

\noindent
We have reported on Table~\ref{tab:GOpt} the standard deviation of 
the Variance Optimal hedging error for different values of coefficient
 $C$ and different choices of rebalancing grids. More precisely, we have considered
 three types of rebalancing grids, for $N=12$ rebalancing dates.
\begin{enumerate}
	\item Equidistant rebalancing dates (corresponding to $\pi^{1,N}$); 
	\item $\pi^{b^*,N}$ where $b^*$ is obtained by minimizing the Variance Optimal hedging error w.r.t. to parameter $b$;
	\item The non parametric optimal grid $\pi^*$ obtained by minimizing the Variance Optimal hedging error w.r.t. the $N$ rebalancing dates.
\end{enumerate}
Notice that in both cases the optimal (parametric and non parametric) grid is estimated by an optimization algorithm based on Newton's method. \\ 
First, one can notice that for any choice of rebalancing grid, the hedging error increases when $C$ decreases. Hence, one can conclude, as expected, that the \textit{degree of incompleteness} increases when the tails of log-returns distribution get heavier.\\
Besides, one can notice that the parametrization~\eqref{eq:pib} of the rebalancing grid seems remarkably relevant since the optimal parametric grid $\pi^{b^\ast}$ achieves similar performances as the optimal non-parametric grid $\pi^{\ast}$. \\
Moreover, we observe that the hedging error can be noticeably reduced by optimizing the rebalancing dates essentially for $C\geq 1$ i.e. around the Gaussian case. In these cases, one can observe on Figure~\ref{fig:Hed:Digital} that the optimal rebalancing grid is noticeably different from the uniform grid since rebalancing dates are much more concentrated near maturity. This confirms the result of~\cite{Gei02} that shows that, in the Gaussian case, taking a non uniform rebalancing grid (corresponding to $b=0.5$) allows to obtain a hedging error with the convergence order for the $L^2$ norm  of $N^{-1/2}$  (up to a log factor) improving the rate $N^{-1/4}$ achieved with a uniform rebalancing grid (i.e. $b=1$), obtained in~\cite{GT01}. 
However, it is interesting to notice that this phenomenon is less pronounced when the tails of the log-returns distribution get heavier. In particular, one can observe on Figure~\ref{fig:Hed:std_b}  that the hedging error gets less sensitive to the rebalancing grid  when $C$ decreases even if the optimal grid seems to get closer to the uniform grid. 

\begin{table}[htbp]
\begin{center}
\begin{tabular}{|c||c|c|c|c|c|}
\hline
 & $C=2$ & $C=1$ & $C=0.2$ & $C=0.14$\\
\hline
\hline
$10\times STD_{VO(\pi^*)}$&  $1.483$ $(30.82)$&  $1.652$ $(34.33)$&  $2.663$ $(54.80)$&  $ 3.017$ $(61.53)$\\
\hline
$10\times STD_{VO(\pi^{b*})}$&  $1.520$ $(31.58)$&  $1.685$ $(35.01)$&  $2.665$ $(54.84)$&  $3.017$ $(61.53)$\\
\hline
$10\times STD_{VO(\pi^{1})}$ &  $1.892$ $(39.32)$&  $1.952$ $(40.56)$&  $2.691$ $(55.38)$&  $3.028$ $(61.76)$\\
\hline
\hline
$V_0(\pi^{1})$&  $0.4903$&  $0.4859$&  $0.4813$&  $0.4812$ \\
\hline
  $V_0(\pi^{*})$&  $0.4903$&  $0.4860$&  $0.4814$&  $0.4813$ \\
\hline
  \hline
$b^*$&  $0.4078$&  $0.4394$&  $0.6106$&  $0.6710$ \\
\hline
\end{tabular}
\end{center}
\caption{{\small Standard deviation of the Variance Optimal hedging
 error  ($\times 10$)  (reported within parenthesis in percent of the
 initial capital  $V_0(\pi^{1})$), initial capitals for $b=1$
and $b = b^*$, optimal grid
 parameters $b^*$, for different choices of parameters $C$  with
  $N=12$ and $K=99$ (Digital option). }}
\label{tab:GOpt}
\end{table}

\begin{figure}[htbp]
\begin{center}
   \epsfig{figure=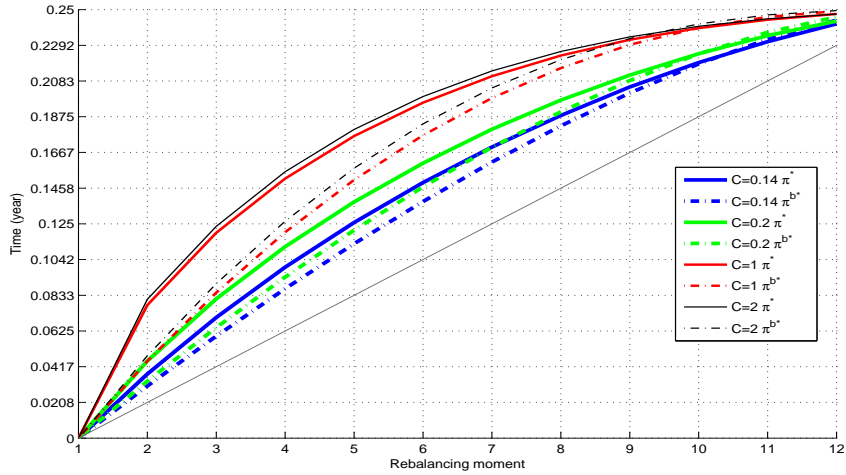,height=7cm, width=13cm}
\end{center}
\caption{{\small Parametric and non parametric optimal rebalancing grids for different choices of parameter $C$ with  $N=12$ and $K=99$ (Digital option). }}
\label{fig:Hed:Digital}
\end{figure}
\begin{figure}[htbp]
\begin{center}
   \epsfig{figure=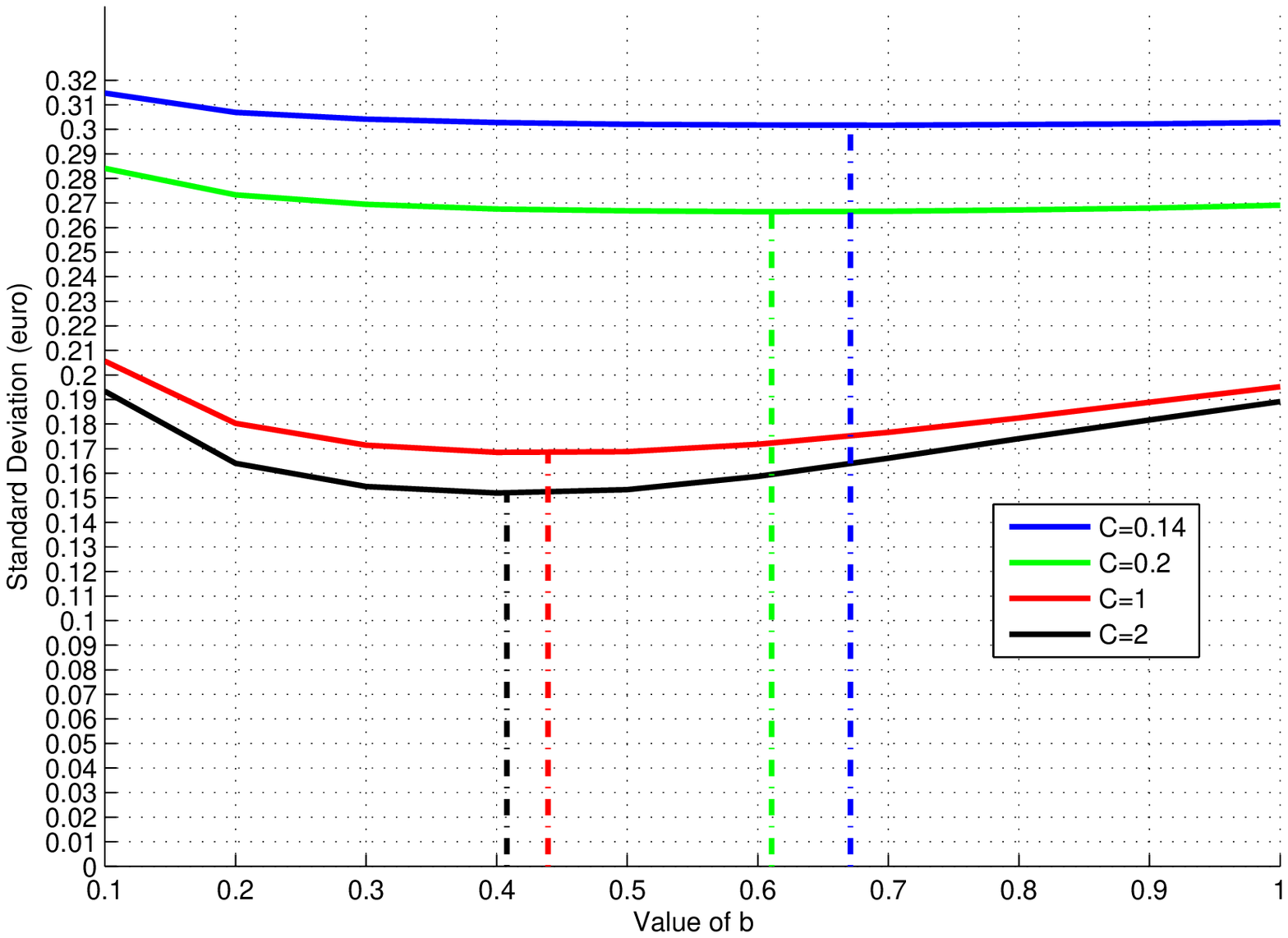,height=7cm, width=13cm}
\end{center}
\caption{{\small Standard deviation of the Variance Optimal hedging error as a function of $b$, for different choices of parameter $C$  ($b^\ast$ being indicated by the dashed line abscissa) with  $N=12$ and $K=99$ (Digital option).}}
\label{fig:Hed:std_b}
\end{figure}


\subsection{The case of electricity forward prices}
\label{electricity}
%
We consider the problem of hedging and pricing a European call,
 with payoff $(F^{T_d}_T-K)_+$, on an electricity forward, with a maturity
 $T=0.25$ of three month. The  maturity $T$ is supposed to be equal to the delivery date of
 the forward contract $T=T_d$. Because of non-storability of electricity, the hedging instrument is the corresponding forward contract.
 Then we set  $S^c_t=F_t^{T}$, where the forward price $F^{T}$ is supposed to follow the  NIG one factor model~\eqref{eq:elecSimple} with $m \equiv 0, \sigma_L = 0$ and  $\sigma_s = \sigma>0$.
This gives
\begin{equation}
\label{eq:model:elec}
S^c_t=e^{X^c_t}\ ,\quad \textrm{where $X^c_t=\int_0^t \sigma 
e^{-\lambda(T-u)}d\Lambda_u$} \quad  \textrm{where $\Lambda$ is a NIG process with 
 }\quad \Lambda_1\,\sim\, {\rm NIG}(\alpha,\beta,\delta,\mu)\ .
\end{equation}
Given $N+1$ discrete dates $0=t_0<t_1<\cdots <t_N=T$,
we consider the discrete process $X = X^N$ where
$X_k = X^c_{t_k}, 0 \le k \le  N$. We denote again by $z\mapsto m(z,k)$
 the cumulant generating function  associated
 with the increment $\Delta X_k=X_k-X_{k-1}$ for 
$k\in\{1,\cdots N\}$. That function and its domain
 can be deduced from
 Lemma~3.24  and Proposition 6.2 in~\cite{GOR}, 
see also \eqref{laplace:nig}. The domain 
$D$ contains 
$\tilde D := [-\frac{\alpha+\beta}{\sigma},
 \frac{\alpha-\beta}{\sigma} ]+ i \R$
 and given for any $z \in \tilde D, k = 0, \ldots, N, $ 
\begin{eqnarray}
\label{q:mWienerNIG}
 m(z,k)
 &=&\mathbb{E} [\exp(z\int_{t_{k-1}}^{t_k} \sigma e^{-\lambda(T-u)}d\Lambda_u)]\nonumber \\ 
 &=&\exp \left (\int_{t_{k-1}}^{t_k} \kappa^\Lambda(z_u
du \right )\,, \quad \textrm{with } z_u=z\sigma  e^{-\lambda(T-u)}\nonumber \\ 
 &=&\exp{\left (\int_{t_{k-1}}^{t_k}  \big [\mu z_u +\delta \big (\sqrt{\alpha^2-\beta^2}-\sqrt{\alpha^2-(\beta+z_u)^2}\big )\big ] du \right )},
\end{eqnarray}%
where $\kappa^\Lambda$ is recalled in formula  \eqref{laplace:nig}.
Hence Assumption~\ref{HypD1} 1. is obviously 
satisfied since $\lambda \neq 0$ and Assumption~\ref{HypD1} 2.
  is verified as soon as $\sigma \le \frac{\alpha-\beta}{2}$; thanks
 to Remark~\ref{RVerAss}, Assumption~\ref{HypD2} is automatically
 verified for the call representation given by Lemma~\ref{Call}. \\
Parameters are  estimated on the same data as in the previous section, with 
\textit{Month-ahead} \textit{base} forward prices of the French Power market in 2007.
  For the distribution of $\Lambda_1$ this yields the following parameters  
$$\alpha=15.81 \,,\  \beta= -1.581\,,\ \delta = 15.57\,,\  \mu= 1.56 \ ,$$
corresponding to  a standard and centered NIG distribution with a skewness of $-0.019$ and excess kurtosis $0.013$.   The estimated annual  short-term volatility and mean-reverting rate are $\sigma=57.47\%$ and $\lambda=3$. 

\noindent
We have reported on Figure~\ref{fig:Hed:PAI}, the standard deviation of the hedging error as a function of the number of rebalancing dates for four types of hedging strategies. 
\begin{itemize}
	\item \textbf{Variance Optimal} strategy (VO) with the \textbf{uniform rebalancing grid (dark line)} and with the \textbf{optimal rebalancing grid $\pi^\ast$ (dark dashed line)}. 
Both variances are computed using formula~\eqref{eq:J0} applied to the process~\eqref{eq:model:elec};
	\item \textbf{Black-Scholes} strategy (BS) implemented at the discrete instants of the \textbf{uniform rebalancing grid (light line)} and of the \textbf{rebalancing grid $\pi^\ast$ (optimal for the Variance Optimal strategy) (light dashed line)}. 
Both variances are computed using Theorem 3.1 of \cite{AH09} extended to
 non-stationary log-returns, to derive a quasi-explicit formula for the variance of the BS hedging error. 
Indeed, in~\cite{AH09}, the authors uses the Laplace transform approach, to derive quasi-explicit formulae for the mean squared hedging error of various discrete time hedging strategies including Black-Scholes delta when applied to L\'evy log-returns models. This extension of this result to the general case when
$X$ is a non-stationary process with independent increments is
given below.

\begin{propo}\label{propoANG}
Let $\upsilon$ be an admissible strategy satisfying
\begin{equation}\label{ConditionANG}
\upsilon_n=\int f^\upsilon(z)_nS^{z-1}_{n-1}\Pi(dz)
\end{equation}
for $n=1,\dots,N$, where $f^\upsilon(z)_n$ is a deterministic function of the complex variable $z$.
Let c be the initial capital; 
the bias and the variance of the hedging error 
 $\epsilon(\upsilon,c) := H - c - \sum_{k=1}^N \upsilon_k
\Delta S_k $ is given by 
\begin{equation}
\label{eq:bias}
\E[\epsilon(\upsilon, c)] =
\int S_0^z \left [
\prod_{k=1}^N m(z,k) -\sum_{k=1}^N f^\upsilon(z)_k (m(1,k)-1) \prod_{l=2}^k m(z,l-1)\right ] \Pi(dz)-c
\end{equation}
\begin{equation}
\label{eq:Error2}
\E[\epsilon(\upsilon, 0)^2] =\int\int 
S_0^{y+z} (v_1(y,z)-v_2(y,z)-v_3(y,z)+v_4(y,z))\Pi(dz)\Pi(dy)\ ,
\end{equation}
where 
\begin{eqnarray*}
v_1(y,z)&=&\prod_{k=1}^N m(y+z,k)\\
v_2(y,z)&=&
\sum_{k=1}^N f^\upsilon(y)_k\, [m(z+1,k)-m(z,k)] \prod_{l=1}^{k-1}m(y+z,l)\prod_{l=k+1}^N m(z,l) \\
v_3(y,z)&=&
\sum_{k=1}^N f^\upsilon(z)_k\, [m(y+1,k)-m(y,k)] \prod_{l=1}^{k-1}m(y+z,l)\prod_{l=k+1}^N m(y,l) \\
v_4(y,z)&=&
\sum_{k=1}^N f^\upsilon (y)_k f^\upsilon(z)_k\, [m(2,k)-2m(1,k)+1] \prod_{l=1}^{k-1} m(y+z,l)  \\
&&+
\sum_{j=2}^N \sum_{k<j} f^\upsilon (z)_k f^\upsilon(y)_j\, [m(y+1,k)-m(y,k)]\prod_{l=1}^{k-1} m(y+z,l)\prod_{l=k+1}^{j-1} m(y,l)(m(1,j)-1)\\
&&+
\sum_{j=2}^N \sum_{k>j} f^\upsilon (z)_k f^\upsilon(y)_j\, [m(z+1,j)-m(z,j)] \prod_{l=1}^{j-1} m(y+z,l)\prod_{l=j+1}^{k-1} m(z,l)(m(1,k)-1)\ .
\end{eqnarray*}

Therefore, the variance of the hedging error is
$$
{\rm Var}(\epsilon(\upsilon ; c)) = {\rm Var}(\epsilon(\upsilon ; 0)) = \E[(\epsilon(\upsilon ; 0)^2] + \E[\epsilon(\upsilon ; 0)]^2\ .
$$
\end{propo}
\begin{proof}
The proof is similar to the one of Theorem 3.1  of \cite{AH09}.
\end{proof}
\begin{remarque}
In the case of Black-Scholes delta hedging strategy 
$$
f^\upsilon(z)_n=z\prod_{k=n}^Nm^{bs}(z,k) \ , \quad \textrm{where} \  m^{bs}(z,k)=\exp\left(-\frac{Var[\Delta X_k]}{2}z+\frac{Var[\Delta X_k]}{2}z^2\right).
$$
\end{remarque}
\end{itemize}

Observing Figure~\ref{fig:Hed:PAI},  one can notice that, as expected, in all cases,  the hedging error decreases when the number of trading dates increases. 
Observing the continuous lines, corresponding to  a uniform rebalancing grid, one can notice the remarkable robustness of the Black-Scholes  strategy. Indeed, in spite of the non Gaussianity of log-returns and the discreteness of the rebalancing grid, the Black-Scholes strategy is still quasi optimal in terms of variance. \\
Besides, in this case, the impact of the choice of the rebalancing grid seems to be more important than the choice of log-returns distribution (Gaussian or Normal Inverse Gaussian).  For instance, using the VO strategy with the optimal rebalancing grid $\pi^\ast$ instead of $\pi^1$ allows to reduce $9\%$  (for $N=10$) of the hedging error standard deviation. 
The BS strategy shows  similar performances  to the VO case,
  when implemented at the rebalancing times $\pi^\ast$. Indeed
\textit{BS optimal rebalancing grid} (in terms of variance)
 appears to be  close to $\pi^\ast$ 
(up to $10^{-4}$).
Moreover, one can observe on Table~\ref{tab:GOptPAI} that here again, the
 parametrization~\eqref{eq:pib} of the rebalancing grid seems to be 
particularly well suited since it achieves minimal hedging errors comparable
 to the one achieved with the nonparametric optimal grid $\pi^\ast$.\\

 Notice that our analysis only considers 
the variance of the hedging error. To obtain the mean square error, one 
should add the bias contribution which is of course zero
 for the variance optimal
 strategy but it is  in general non negligible for the Black-Scholes
 strategy. In particular, we can observe that this bias term varies 
strongly with the parameters of the NIG distribution.

For instance, for $N=2$ uniform rebalancing dates, replacing parameter $\beta$ 
by $-\beta$ increases the bias (defined as~\eqref{eq:bias}, with
initial capital $c=V_0^{\rm BS})$ 
from -0.04 to 4.45. Moreover, one should also observe that the drift
 and the skewness of log-returns 
also impact the standard deviation of the BS hedging error. 
Changing again $\beta$ by $-\beta$ implies an increase of the
log-returns
 expectation 
(resp. skewness) from 0 to 3.12 (resp.  from -0.02 to 0.02) which induces an 
increase of the standard deviation of the BS hedging error 
 from $4.91$ to $5.92$, whereas the standard deviation 
of the VO hedging error decreases from 4.83 to 2.10.

\begin{figure}[htbp]
\begin{center}
   \epsfig{figure=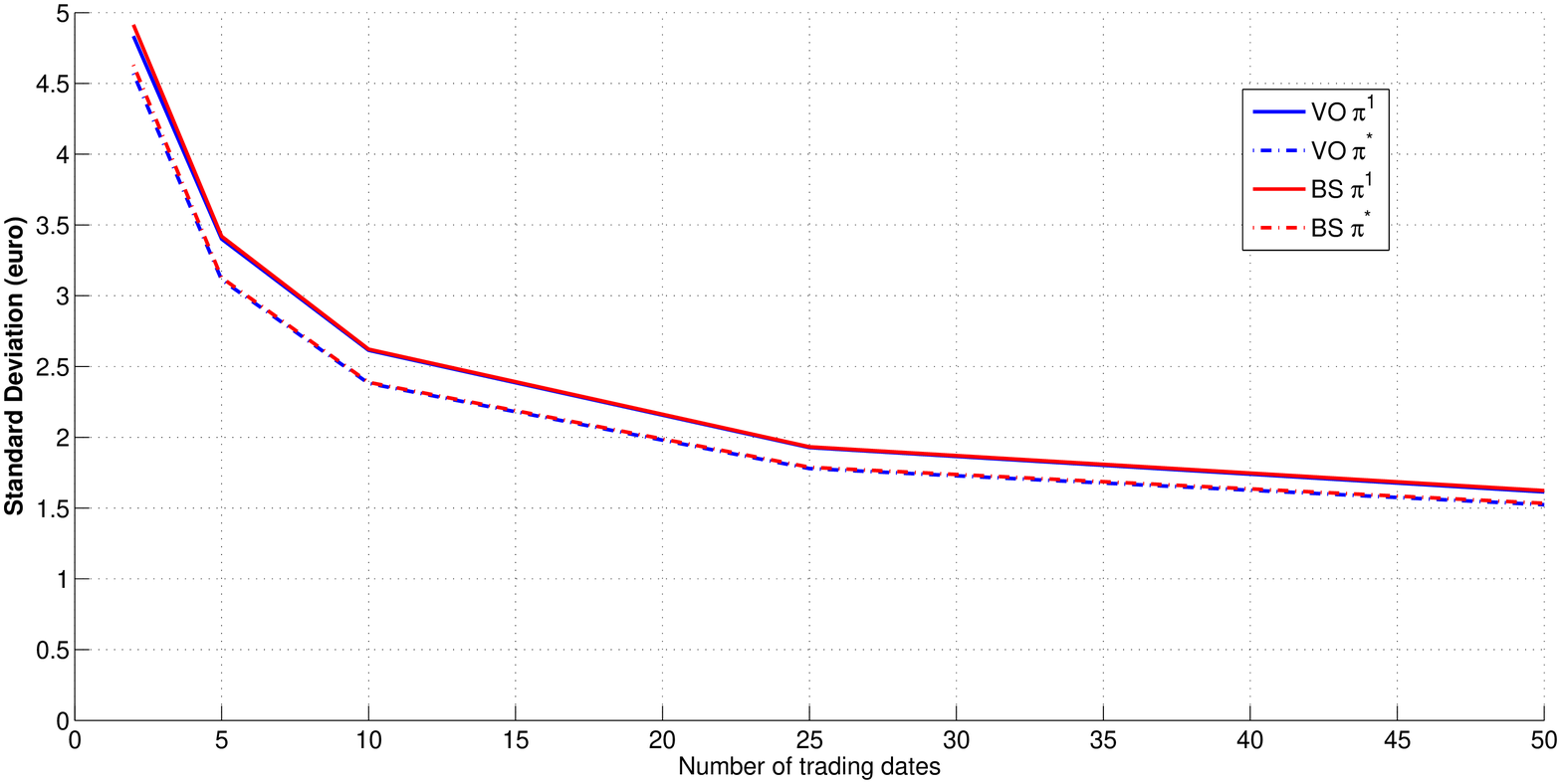,height=7cm, width=13cm}
\end{center}
\caption{{\small Standard deviation of the hedging errors as a function of the number of rebalancing dates $N$, for $K=99$ (Call option). }}
\label{fig:Hed:PAI}
\end{figure}
\begin{table}[htbp]
\begin{center}
\begin{tabular}{|c||c|c|c|c|c|}
\hline
  & $N=2$ & $N=5$ & $N=10$ & $N=25$ & $N=50$\\
\hline
\hline
 $STD_{VO(\pi^*)}$&  $4.5683 $ $(53.23)$&  $3.1129$ $(36.10)$&  $2.3807$ $(27.56)$&  $ 1.7790$  $(20.57 )$&  $ 1.5233 $ $(17.61)$\\
\hline
 $STD_{VO(\pi^{b*})}$ &  $4.57167$ $(53.27)$&  $3.1550$ $(36.59)$&  $2.4186$ $(28.00)$&  $1.8023  $ $(20.84)$ &  $ 1.5354 $ $(17.75)$\\
 \hline
 $STD_{VO(\pi^{1})}$ &  $4.8331 $ $(56.32)$&  $3.4012$ $(39.44)$&  $ 2.6154$ $(30.28)$&  $ 1.9275 $ $(22.29)$ &  $ 1.6145$ $(18.66)$\\
  \hline
    \hline
 $STD_{BS(\pi^{1})}$ &  $4.9137 $ $(57.26)$&  $3.4196$ $(39.66)$&  $ 2.6217$ $(30.35)$&  $ 1.9329 $  $(22.35)$&  $ 1.6231$ $(18.76)$\\
  \hline                 
 $STD_{BS(\pi^{*})}$ &  $4.6291$ $(53.94)$&  $3.1273$ $(36.27)$&  $ 2.3884$ $(27.65)$&  $ 1.7886 $ $(20.68)$ &  $ 1.5344$ $(17.74)$\\
 \hline
  \hline
$V_0(\pi^{1})$&  $8.5818$&  $8.6232$&  $8.6380$&  $8.6469$ &  $8.6499$\\
 \hline
 $V_0(\pi^{*})$&  $8.5895$&  $8.6275$&  $8.6406$&  $8.6493$ &  $8.6531$\\
 \hline
  \hline
  \hline   
$b^*$&  $0.5917$&  $0.6298$&  $0.6284$&  $0.6203$ &  $0.6172$\\
 \hline
\end{tabular}
\end{center}
\caption{{\small Standard deviation of the Variance Optimal hedging 
error (reported within parenthesis in percent of the initial 
capital $V_0(\pi^{1})$), initial capitals, optimal grid parameters
$b^*$, for different choices of rebalancing dates $N$ (Call option). }}
\label{tab:GOptPAI}
\end{table}
%



To analyze the impact of the rate of volatility increase on the
 optimal rebalancing grid, we have computed the hedging error standard 
deviation for several values of parameter $\lambda$ choosing the
 corresponding  volatility parameter $\sigma$ such that
 $Var(X_T)=\frac{\sigma^2}{2\lambda} (1-e^{-2\lambda T})$ is fixed.
  The resulting pairs $(\lambda,\sigma)$ are reported on
 Table~\ref{tab:lambdasigma}. Coupling those parameters allows us
 to obtain comparable options for different parameters $\lambda$; at
 least this ensures a fixed initial capital in the BS framework
 (with $V_0^{BS}=8.7037$). 
On Figure~\ref{fig:STD:lamabda1}, we have reported the optimal grid parameter $b^\ast$ minimizing the standard deviation of the VO hedging error for different values of $\lambda$. 
As expected, when $\lambda$ increases, i.e. when the volatility increases more rapidly near the maturity, then $b^\ast$ decreases indicating that the optimal rebalancing dates concentrate near the maturity. 	On Figure~\ref{fig:STD:lambda2}, one can observe that the hedging error increases with $\lambda$ even when the rebalancing dates are optimized.  However, optimizing the rebalancing dates allows to reduce noticeably the hedging error, specifically for high values of $\lambda$. For instance, it  allows to reduce $7.5\%$ of the error standard deviation when $\lambda=3$ and  $17.9\%$ when $\lambda=9$. 

\begin{table}[htbp]
\begin{center}
\begin{tabular}{|c||c|c|c|c|c|c|}
\hline
 $\lambda$&    $1$&  $2$&  $3$&  $6$ &  $9$\\
\hline              
 $\sigma$&   $0.4662$&  $0.5202$&  $0.5747$&  $0.7349$&  $0.8823$ \\
\hline
\end{tabular}
\end{center}
\caption{{\small Short term volatility $\sigma$  (s.t. $Var(X_T)=\frac{\sigma^2}{2\lambda} (1-e^{-2\lambda T})$ is fixed) 
 for different values of parameter $\lambda$ with $N=10$ and $K=99$ (Call option).}}
\label{tab:lambdasigma}
\end{table}
\begin{figure}[htbp]
\begin{center}
   \epsfig{figure=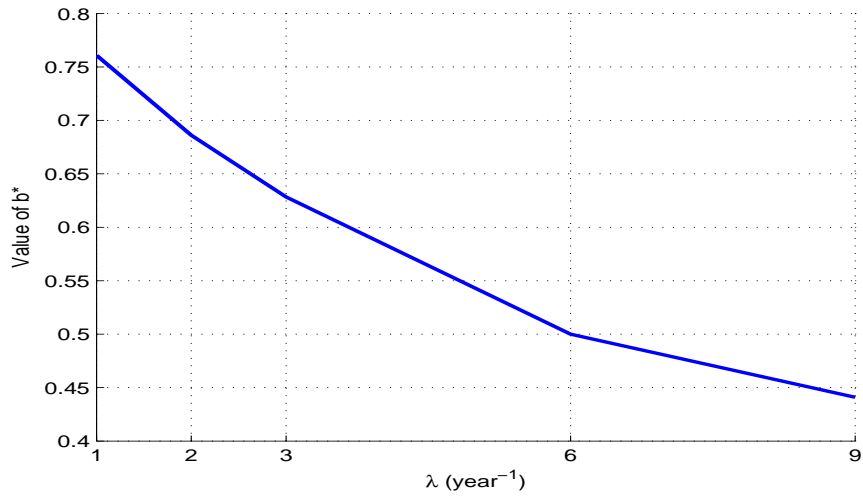,height=7cm, width=13cm}
\end{center}
\caption{{\small Optimal rebalancing grid parameter $b^*$ as a function of $\lambda$ for $K=99$ and $N=10$ (Call option).}}
\label{fig:STD:lamabda1}
\end{figure}
\begin{figure}[htbp]
\begin{center}
   \epsfig{figure=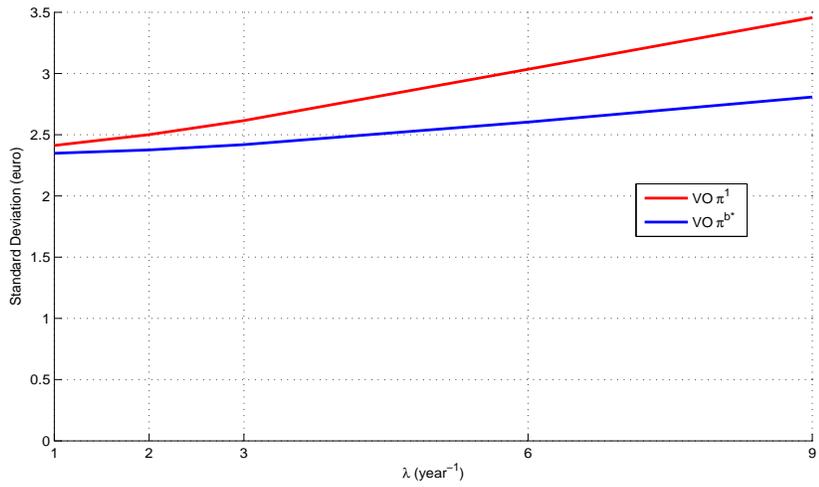,height=7cm, width=13cm}
\end{center} 
\caption{{\small Standard deviation of the hedging error as a function of  $\lambda$ for $K=99$ and $N=10$  (Call option). }}
\label{fig:STD:lambda2}
\end{figure}

\newpage

\section{Appendix}

\subsection*{A: A general convergence theorem 
for FS decompositions}

\begin{propo}\label{AP2}
Let $(H^\ell)_{\ell\in \N \cup \{+\infty\}}$ be a sequence of r.v. in $\shl^2(\Omega,\shf_N)$. Let 
\begin{equation}\label{EP2}
\left \{ 
\begin{array}{ccc}
H_n^\ell&=&H_0^\ell+\sum_{i=1}^n \xi^\ell_i\Delta S_i+L_n^\ell\\
H_N^\ell &=& H^\ell
\end{array}
\right.
\end{equation}
be the FS-decomposition of $H^\ell$. Suppose that $H^\ell \rightarrow H^\infty$ in $\shl^2(\Omega)$. Then, for $\ell \rightarrow  +\infty$,
\begin{enumerate}
\item $H_0^\ell\rightarrow H_0^\infty$ in $\shl^2(\shf_0)$;
\item $\xi_n^\ell\rightarrow \xi_n^\infty$ in probability for any
 $n\in \{1,\dots,N\}$;
\item $L_N^\ell\rightarrow L_N^\infty$ in $\shl^2(\Omega)$.
\end{enumerate}
\end{propo}

\begin{proof}
For $n\in \{1,\dots,N\}$, $ \ell \in \N \cup \{+\infty\}$
 we have
\begin{equation}\label{EP3}
H_n^\ell=H^\ell_{n-1}+\xi^\ell_n\Delta S_n+\Delta L^\ell_n.
\end{equation}
For technical reasons we set $\xi^\ell_{N+1}:=0$ and
$L^\ell_{N+1}:=L^\ell_{N}$. The result will follow if for every $n\in
\{0,\dots,N\}$, 
 for $\ell \rightarrow + \infty$
we have
\begin{enumerate}
\item $H_n^\ell\rightarrow H_n^\infty$ in $\shl^2$,
\item $\E\left[\left(\Delta S_{n+1}\right)^2\left(\xi^\ell_{n+1}-
\xi^\infty_{n+1}\right)^2\right] \rightarrow 0$,
\item $L^\ell_{n+1}\rightarrow L^\infty_{n+1}$ in $\shl^2$.
\end{enumerate}

We will prove 1.,2. and 3. by backward induction on $n\in \{0,\dots,N\}$ starting from $n=N$. The step $N$ of the induction is constituted by the assumption, in particular 1. and 3. are  verified by assumption 
and 2. is trivially verified. \\
Suppose that 1.,2. and 3. hold for some $n\in \{1,\dots,N\}$, we will prove
  their validity for the integer $n-1$. First,
1. implies  that $\E\left[H_n^\ell\vert \shf_{n-1}\right]\rightarrow
_{\ell  \rightarrow +\infty}
 \E\left[H_n^\infty\vert \shf_{n-1}\right]$ in $\shl^2(\Omega)$.
 We continue taking the conditional expectation with respect to $\shf_{n-1}$ in \eqref{EP3}. This gives
\begin{equation}\label{EP4}
\E\left[H_n^\ell\vert \shf_{n-1}\right]=H^\ell_{n-1}+\xi^\ell_n\Delta A_n.
\end{equation}
The difference between \eqref{EP3} and \eqref{EP4} gives
$$
H_n^\ell-\E\left[H_n^\ell\vert \shf_{n-1}\right]=\xi^\ell_n\Delta M_n +
 \Delta L_n^\ell, \ \ell\in \N \cup \{+\infty\}.
$$
Consequently
$$ H_n^\ell-H_n^\infty=\E\left[H_n^\ell-H_n^\infty\vert \shf_{n-1}\right] 
= \left(\xi^\ell_n-\xi^\infty_n\right)\Delta M_n +\Delta(L_n^\ell-L_n^\infty).
$$  
So
\begin{equation}\label{PM6}
\E\left[\left(H_n^\ell-H_n^\infty\right)-\E\left[H_n^\ell-H_n^\infty\vert \shf_{n-1}\right]^2\right]=\E\left[\left(\xi^\ell_n-\xi^\infty_n\right)^2\left(\Delta M_n\right)^2\right]+\E\left[\Delta \left(L_n^\ell-L_n^\infty\right)\right]^2;
\end{equation}
in fact
\begin{equation*}
\E\left(\left(\xi^\ell_n-\xi^\infty_n\right)\Delta M_n \Delta 
\left(L_n^\ell-L_n^\infty\right)\right)=
\E\left((\xi^\ell_n-\xi^\infty_n) 
\E\left((\Delta M_n)\Delta (L_n^\ell-L_n^\infty)
\vert \shf_{n-1} \right)  \right)= 0,
\end{equation*}
because $M.\left(L^\ell-L^\infty\right)$ is a martingale. Since the
 left-hand side of \eqref{PM6} converges to zero when $\ell 
\rightarrow \infty$, it follows that
\begin{eqnarray}\label{PM7}
\E\left[\left(\xi^\ell_n-\xi^\infty_n\right)^2\left(\Delta M_n\right)
^2\right]&\rightarrow_{\ell\rightarrow \infty}&0\\
\E\left[\Delta \left(L_n^\ell-L_n^\infty\right)\right]^2&
\rightarrow_{\ell\rightarrow \infty}&0.
\nonumber
\end{eqnarray}
This shows 2. and 3. of the $(n-1)$-step of the backward induction. It remains to show item 1. By \eqref{EP3}, we have
\begin{equation*}H_{n-1}^\ell-H_{n-1}^\infty = H_{n}^\ell-H_{n}^\infty
  - \Delta S_n \left(\xi^\ell_n-\xi^\infty_n\right)-
\Delta \left(L_n^\ell-L_n^\infty\right).
\end{equation*}
Since $H_{n}^\ell-H_{n}^\infty$ and $\Delta
\left(L_n^\ell-L_n^\infty\right)$ 
converge to zero in $\shl^2$, it remains to show that 
$ \Delta S_n \left(\xi^\ell_n-\xi^\infty_n\right)
\rightarrow_{\ell\rightarrow \infty}0$ in $\shl^2(\Omega)$ when $\ell\rightarrow \infty$. Now $\Delta M_n \left(\xi^\ell_n-\xi^\infty_n\right)
\rightarrow_{\ell\rightarrow \infty}0$ in $\shl^2(\Omega)$ and so by
 \eqref{PM7}
 we only have to prove that
\begin{equation}\label{PM8}
\E\left[\left(\xi^\ell_n-\xi^\infty_n\right)^2\left(\Delta A_n\right)^2\right]\rightarrow_{\ell\rightarrow \infty}0.
\end{equation}
By the (ND) condition and item 1. of Remark \ref{remarqueRD}, we have
\begin{eqnarray*}
\left(\Delta A_n\right)^2&=&  
\left (\E\left(\Delta S_n \vert \shf_{n-1}\right) \right)^2 
\le \delta \E\left(\left(\Delta S_n\right)^2\vert \shf_{n-1}\right)
\\
&=&\delta\left((\Delta A_n)^2+\E\left[\left(\Delta M_n\right)^2\vert \shf_{n-1}\right]\right).
\end{eqnarray*}
Consequently
\begin{eqnarray*}
\left(\Delta A_n\right)^2&\leq&\frac{\delta}{1-\delta}\E\left[\left(\Delta M_n\right)^2\vert \shf_{n-1}\right].
\end{eqnarray*}
So the left-hand side of \eqref{PM8} is bounded by 
$$
\frac{\delta}{1-\delta}\E\left[\left(\xi^\ell_n-\xi_n^\infty\right)^2\left(\Delta M_n\right)^2\right]\rightarrow_{\ell\rightarrow \infty} 0.
$$
The result is finally established.
\end{proof}

\subsection*{B: The Normal Inverse Gaussian distribution}
\label{Appendix:NIG}
The Normal Inverse Gaussian (NIG) distribution is a specific subclass of the Generalized Hyperbolic family introduced by Barndorff--Nielsen in 1977, see
 for instance~\cite{bnh}.  The density of a Normal Inverse Gaussian distribution of parameters $(\alpha,\beta,\delta,\mu)$ is given by  
\begin{equation}
\label{dens_nig}
f_{NIG}(x;\alpha, \beta, \delta, \mu)= \frac{\alpha}{\pi} \exp \big ( \delta \sqrt{\alpha^2-\beta^2}+\beta (x-\mu)\big ) \frac{K_1\big ( \alpha \delta \sqrt {1+(x-\mu)^2/\delta^2}\big )}{\sqrt{1+(x-\mu)^2/\delta^2}}\ ,\quad \textrm{for any}\ x\in\mathbb{R}\ ,
\end{equation}
where $K_1$ denotes the Bessel function of the third type with index 
 $1$ and where the  parameters  are such that  $\delta> 0 $, $\alpha >0$ 
 and $\alpha > \vert \beta\vert$. 
Afterwards,  NIG$(\alpha, \beta, \delta,\mu)$ will denote the Normal Inverse Gaussian  distribution of parameters $(\alpha, \beta, \delta,\mu)$. 

A useful property of the NIG distribution is its stability  under convolution i.e. 
$$
NIG(\alpha, \beta, \delta_1,\mu_1)\ast NIG(\alpha, \beta, \delta_2,\mu_2)=NIG(\alpha, \beta, \delta_1+\delta_2,\mu_1+\mu_2)\ .
$$
This property shared with the Gaussian distribution allows to simplifies many computations. 

If  $X$ is a NIG($\alpha,\beta,\delta,\mu$) random variable then  for any $a\in \mathbb{R}^+$ and $b\in \mathbb{R}\,,$ $Y=aX+b$ is also a NIG random variable with parameters  ($\alpha /a,\beta /a,a\delta,a\mu +b$). 

The mean and the variance associated to a NIG$(\alpha, \beta, \delta,\mu)$ random variable $X$ are given by, 
\begin{equation}
\label{eq:mstd_nig}
\mathbb{E} X =\mu +\frac{\delta \beta}{\gamma}\ ,\quad \textrm{Var} X=\frac{\delta \alpha^2}{\gamma^3}\ ,\quad\textrm{with}\quad \gamma=\sqrt{\alpha^2-\beta^2}\ . 
\end{equation}

The characteristic function of the NIG distribution is given by  $\exp(\Psi_{NIG})$ where $\Psi_{NIG}$ verifies 
\begin{equation}
\label{eq:psi:nigb}
\Psi_{NIG}(u)=\log \mathbb{E}\big [ \exp(iuX)\big ]=i\mu u +\delta (\sqrt{\alpha^2-\beta^2}-\sqrt{\alpha^2-(\beta+iu)^2})\quad \textrm{for any}\quad u\in \mathbb{R} \ .
\end{equation}

The moment generating function of the NIG distribution is particularly simple,  
\begin{equation}
\label{laplace:nig}
\kappa^\Lambda(z) = \kappa^\Lambda_{NIG}(z)=\log  \mathbb{E} [\exp(zX)]=\mu z +\delta \big (\sqrt{\alpha^2-\beta^2}-\sqrt{\alpha^2-(\beta+z)^2}\big )\ , \quad \textrm{for} \ Re(z)\in[-(\alpha+\beta);\alpha-\beta]\ .
\end{equation}

The L\'evy measure of the NIG distribution is given by  
\begin{equation}
\label{eq:levy:nig}
F_{NIG}(dx)=e^{\beta x}\frac{\delta \alpha}{\pi \vert x\vert }K_1(\alpha \vert x\vert)\,dx\quad \textrm{for any}\ x\in\mathbb{R}\ .
\end{equation}
Notice that the L\'evy measure does not depend on parameter $\mu$. 
%
%


\bigskip
{\bf ACKNOWLEDGEMENTS:} The authors are  grateful to the anonymous Referee 
for her$\backslash$his  stimulating remarks and comments which allowed
them to considerably improve the first version of the paper.

The first named author was partially founded
by Banca Intesa San Paolo.
The research of the third named author was partially 
supported by the ANR Project MASTERIE 2010 BLAN-0121-01.


\begin{thebibliography}{9}






\bibitem{AH10} Angelini, F. and Herzel., S. (2010).  \textit{Explicit formulas for the minimal hedging strategy in a martingale case}. Decisions in Economics and Finance Vol {\bf 33}(1).


\bibitem{AH09} Angelini, F. and Herzel., S. (2009).  \textit{Measuring the error of dynamic hedging: a Laplace transform approach}. Computational Finance Vol {\bf 12}(2).

\bibitem{bnh} Barndorff-Nielsen, O.E. and Halgreen, C. (1977).
 \textit{Infinite divisibility of the hyperbolic and generalized inverse 
Gaussian distributions} Zeitschrift f\"ur Wahrscheinlichkeitstheorie und verwandte Gebiete, Vol. {\bf 38}, 309-312.

\bibitem{BarndorffNielson} Barndorff-Nielsen,
  O.E. (1998). \textit{Processes of normal inverse Gaussian type},
  Finance and Stochastics {\bf 2}, 41-68.

\bibitem{Benth-Kallsen} Benth, F. E., Kallsen, J. and Meyer-Brandis,
  T. (2007).
  \textit{A non-Gaussian Ornstein-Uhlenbeck process 
for electricity spot price modeling and derivatives pricing},
 Applied Mathematical Finance, \textbf{14}(2), 153-169.


\bibitem{nunno} Benth, F. E., Di Nunno, G.,  
L{\o}kka, A., {\O}ksendal, B. and
 Proske, F. (2003).
{\it Explicit representation of the minimal variance portfolio in 
markets driven by L\'evy processes.} 
Conference on Applications of Malliavin Calculus in Finance 
(Rocquencourt, 2001).
Mathematical Finance {\bf 13}, no. 1, 55--72. 




\bibitem{BS03} Benth, F.E. and Saltyte-Benth, J. (2004).
 \textit{The normal inverse Gaussian distribution and spot price 
modeling in energy markets}, 
International journal of theoretical and applied finance, Vol. {\bf 7}(2), 177-192.


\bibitem{bertsimas}
Bertsimas, D., Kogan, L.  and Lo, A. W. (2001).
{\it Hedging derivative securities and incomplete markets: 
an $\epsilon$-arbitrage approach},  Oper. Res.  {\bf 49},  no. 3.




\bibitem{cerny04} $\check{C}$ern$\grave{y}$  A. (2004).
{\it Dynamic Programming and Mean-Variance Hedging in Discrete Time},
 Applied Mathematical Finance {\bf 11}(1), 1-25.


\bibitem{cerny07} $\check{C}$ern$\grave{y}$  A. (2007).
{\it Optimal Continuous-Time Hedging with Leptokurtic Returns},  
Mathematical Finance, Vol. {\bf 17}(2), pp. 175-203. 

\bibitem{CK08} $\check{C}$ern$\grave{y}$, A. and Kallsen, J. (2007).
\textit{On the structure of general man-variance hedging strategies},
 The Annals of probability, Vol. {\bf 35} N. 4, 1479-1531.

\bibitem{CK09} $\check{C}$ern$\grave{y}$, A. and Kallsen, J.
 (2009). 
{\it Hedging by sequential regressions revisited.} 
Math. Finance {\bf 19}, no. 4, 591--617. 





\bibitem{oudjaneCollet} Collet, J.,  Duwig D. and  Oudjane N. (2006). 
\textit{Some non-Gaussian models for electricity spot prices},
 In Proceedings of the 9th International Conference on Probabilistic
 Methods Applied to Power Systems 2006.

\bibitem{livreTankovCont}  Cont, R. and Tankov, P. (2003).
 \textit{Financial modeling with Jump Processes} Chapman \& Hall / CRC Press.

\bibitem{ContTankov}  Cont, R., Tankov, P.
and  Voltchkova, E. (2007).
\textit{Hedging with options in models with jumps}.
  Stochastic analysis and applications,  197--217, 
Abel Symp., 2, Springer, Berlin.



\bibitem{CRR79} Cox, J. C., Ross, S. A. and Rubinstein, M.  (1979). 
\textit{Option Pricing: A Simplified Approach}. 
 Journal of Financial Economics {\bf 7}: 229-263. 




\bibitem{DGK09} Denkl, S., Goy, M., Kallsen, J., Muhle-Karbe, J. 
 and Pauwels, A. (2009). \textit{On the performance of 
delta-hedging strategies in exponential L\'evy models}. Preprint.

\bibitem{duffierich} Duffie, D. and Richardson H.R. (1991).
Mean-variance hedging in continuous time.
Ann. Appl. Probab. {\bf 1}, no. 1, 1--15. 







\bibitem{FS89} F\"ollmer, H. and Schweizer, M. (1989).
{\it Hedging by Sequential Regression: An Introduction to the 
Mathematics of Option Trading}.
The ASTIN Bulletin {\bf 18}, 147--160.


\bibitem{FS91} F\"ollmer, H. and Schweizer, M. (1991).
{\it Hedging of contingent claims under incomplete information.}
  Applied stochastic analysis (London, 1989),  389-414,
 Stochastics Monogr., {\bf 5}, Gordon and Breach, New York.



\bibitem{Gei02} Geiss, S. (2002). \textit{Quantitative approximation of
 certain stochastic integrals}. Stoch. Stoch.
Rep., {\bf 73}(3-4):241-270.

\bibitem{GeiGei04} Geiss, C. and  Geiss, S. (2004).
\textit{On approximation of a class of stochastic integrals and 
interpolation.}  Stoch. Stoch. Rep.  {\bf 76}, no. 4, 339--362. 

\bibitem{GeiGo10} Geiss S. and Gobet E. (2011).
\textit{Fractional smoothness and applications in finance}. 
Preprint arXiv:1004.3577v1. To appear in 
AMAMEF book, G. Di Nunno and B. {\O}ksendal Eds.

\bibitem{GM09} Gobet, E. and Makhlouf, A. (2010).  
\textit{The tracking error rate of the
Delta-Gamma hedging strategy}.  To appear: Mathematical finance.
 Available at http://hal.archives-ouvertes.fr/hal-00401182/fr/.

\bibitem{GLP98} Gourieroux, C., Laurent, J.-P. and Pham, H. (1998).
Mean-variance hedging and num\'eraire.  
Math. Finance {\bf 8}, no. 3, 179--200.
 


\bibitem{GT01} Gobet, E. and Temam, E. (2001). \textit{Discrete time hedging errors for options with irregular
pay-offs}. Finance and Stochastics {\bf 5}(3):357--367.



\bibitem{goll-ruschendorf} Goll, T. and Ruschendorf, L. (2002).
 \textit{Minimal distance martingale measures and optimal portfolios
   consistent with observed market process}, Stochastic Processes and
 Related Topics,
 {\bf 8} 141-154.

\bibitem{GOR} Goutte, S., Oudjane, N. and Russo, F. (2009).
\textit{Variance Optimal Hedging for continuous time processes 
with independent increments and applications}.
Preprint HAL inria-00437984,
http://fr.arxiv.org/abs/0912.0372.



\bibitem{Ka06} Hubalek, F., Kallsen, J. and Krawczyk, L.(2006).
 \textit{Variance-optimal hedging for processes with stationary
   independent increments}, The Annals of Applied Probability, Vol.
 {\bf 16}, Number 2, 853-885.






\bibitem{KP09b} Kallsen, J.  and Pauwels, A. (2011). 
\textit{Variance-optimal hedging for time-changed L\'evy
  processes,} 
 Appl. Math. Finance  18,  no. 1, 1-28.

\bibitem{KP09} Kallsen, J.  and Pauwels, A. (2010). 
\textit{Variance-optimal hedging in general affine stochastic
 volatility models},
 Advances in Applied Probability {\bf 20}  no. 1, 83-105.

\bibitem{KMS09} Kallsen, J.,  Muhle-Karbe, J., Shenkman, N. and Vierthauer, R. (2009). \textit{Discrete-time variance-optimal hedging in affine stochastic volatility models}. In R. Kiesel, M. Scherer, and R. Zagst, editors, Alternative Investments and Strategies. World Scientific, Singapore.




 













\bibitem{RS97} Rheinl\"ander, T. and Schweizer, M. (1997). 
{\it On $L^2$-projections on a space of stochastic integrals,}
  Ann. Probab.  {\bf 25},  no. 4, 1810--1831.


\bibitem{RU87} Rudin, W. (1987).
 \textit{Real and complex analysis}, third edition. New York: McGraw-Hill.


\bibitem{Schal94} Sch\"al, M. (1994).
 \textit{On quadratic cost criteria for options hedging}, 
Mathematics of Operations Research {\bf 19}, 121-131.


 
\bibitem{S94} Schweizer, M. (1994). \textit{Approximating random
    variables by stochastic integrals}, The Annals of Probability
  Vol. {\bf 22}, 1536-1575.

\bibitem{S95} Schweizer, M. (1995). \textit{On the minimal martingale measure 
and the F\"ollmer-Schweizer decomposition}, 
Stochastic Analysis and Applications,
 {\bf 13},  no. 5, 573--599.
 

\bibitem{S95bis} Schweizer, M. (1995).
 \textit{Variance-optimal hedging in discrete time}, Mathematics of
 Operations Research {\bf 20}, 1-32.

\bibitem{S01} Schweizer, M. (2001).  \textit{A guided tour through quadratic
 hedging approaches.  Option pricing, interest rates and risk management},  
 538-574, Handb. Math. Finance, Cambridge Univ. Press, Cambridge. 

%
%






\end{thebibliography}
\end{document}